\documentclass[12pt]{amsart}
\usepackage[margin=1in]{geometry}
\usepackage{amsfonts}
\usepackage{amssymb}
\usepackage[dvips]{graphics}
\usepackage{epsfig}
\usepackage{euscript}
\usepackage{color}

 \newtheorem{thm}{Theorem}[section]
 \newtheorem{cor}[thm]{Corollary}
 \newtheorem{lemma}[thm]{Lemma}
 \newtheorem{prop}[thm]{Proposition}
 \theoremstyle{definition}
 
 \theoremstyle{remark}
 \newtheorem{rem}[thm]{Remark}
 \newtheorem*{ex}{Example}
 \newtheorem{assumption}[thm]{Assumption}
 \numberwithin{equation}{section}

\newcommand{\R}{{\mathord{\mathbb R}}}
\newcommand{\C}{{\mathord{\mathbb C}}}
\newcommand{\Z}{{\mathord{\mathbb Z}}}

\newcommand{\E}{{\mathord{\mathbb E}}}

\newcommand{\be}{\begin{equation}}
\newcommand{\ee}{\end{equation}}
\newcommand{\bea}{\begin{eqnarray}}
\newcommand{\eea}{\end{eqnarray}}

\def\idty{{\mathchoice {\mathrm{1\mskip-4mu l}} {\mathrm{1\mskip-4mu l}} %
{\mathrm{1\mskip-4.5mu l}} {\mathrm{1\mskip-5mu l}}}}

\definecolor{Plum}{rgb}{.5,0,1}

\newcommand\mat[4]{\left(\begin{array}{cc} #1 & #2 \\ #3 & #4 \end{array}\right)}
\newcommand\eq[1]{(\ref{#1})}
\newcommand{\Tr}{\operatorname{{\rm Tr}}}
\numberwithin{equation}{section}

\begin{document}

\title[Area Law for Disordered Oscillator Systems]{An Area Law for the Bipartite Entanglement\\[5pt] of Disordered Oscillator Systems}

\author[B. Nachtergaele]{Bruno Nachtergaele$^1$}
\thanks{B.\ N.\ was supported in part by NSF grant DMS-1009502}
\address{$^1$ Department of Mathematics\\ University of California, Davis\\ Davis, CA 95616, USA}
\email{bxn@math.ucdavis.edu}

\author[R. Sims]{Robert Sims$^2$}
\thanks{R.\ S.\ was supported in part by NSF grants DMS-0757424 and DMS-1101345}
\address{$^2$ Department of Mathematics\\
University of Arizona\\
Tucson, AZ 85721, USA}
\email{rsims@math.arizona.edu}

\author[G. Stolz]{G\"unter Stolz$^3$}
\thanks{G.\ S.\ was supported in part by NSF grant DMS-1069320.}
\address{$^3$ Department of Mathematics\\
University of Alabama at Birmingham\\
Birmingham, AL 35294 USA}
\email{stolz@math.uab.edu}

\date{}

%\vspace{.3truein}
%\centerline{\bf Abstract}\mdeskip
%
\begin{abstract}
We prove an upper bound proportional to the surface area
for the bipartite entanglement of the ground state and thermal 
states of harmonic oscillator systems with
disorder, as measured by the logarithmic negativity. Our assumptions
are satisfied for some standard models that are almost surely gapless in
the thermodynamic limit.
\end{abstract}

\maketitle

%%%%%%%%%%%%%%%%%%%%%%%%%
%
%       Introduction
%
%
%%%%%%%%%%%%%%%%%%%%%%%%%%%%%

\section{Introduction} \label{sec:intro}

On the one hand, quantum computation and information processing depends in an essential way
on entangled quantum states and, on the other hand,  it is entanglement that is the most serious
limiting factor in the numerical simulation of many-body quantum systems. Both facts provide strong
motivation to study entanglement of the ground states and equilibrium states of important systems
such as oscillator lattice models. Furthermore, it was noted in \cite{Bombelli1986} that entanglement
implies a quantum contribution to black hole entropy. Consequently, the entanglement of
ground states and temperature states of extended systems has been intensely studied, with
establishing an area law bound as one of the main goals.

Vidal and Werner \cite{Vidal2002} made the case for using the {\em logarithmic negativity}
as a measure of entanglement and they showed how it can be computed in a number
of cases, including Gaussian states for bosonic systems. The logarithmic
negativity provides an upper bound for the widely used notions of entropy of entanglement
in the case of pure states and the distillable entanglement in the case of mixed states.
Following the ideas of Vidal and Werner, the logarithmic negativity has been calculated
or estimated for a number of deterministic bosonic systems, primarily lattice systems of
coupled harmonic oscillators \cite{Audenaert2002,Plenio2005,Cramer2006, Cramer2007}.
More recently, the logarithmic negativity has also been used to quantify entanglement in
relativistic quantum field theories \cite{Calabrese2012}.

In this paper we prove an upper area law bound for the logarithmic negativity for a 
class of disordered harmonic lattice models. The previous results of this type that are 
not restricted to one dimension, for either deterministic or disordered systems, all assumed 
the existence of a spectral gap above the ground state. For disordered systems, it is often 
no longer natural to suppose that there is a positive lower bound for the spectral gap uniform 
in the volume. Therefore, we will not make this assumption in this work.  In the special case of
one dimension an alternative approach may be possible following \cite{Brandao2012}.

For this paper we have opted to include sufficient background information so that we can 
give complete proofs of results that often have appeared in a more restricted setting in the
literature. E.g., in the Appendix we discuss the notion of partial transpose for operators on
a the tensor product of two separable Hilbert spaces. We also included a proof of the result 
by Vidal and Werner \cite{Vidal2002}  that the logarithmic negativity is an upper bound for
the entanglement entropy of a pure state in that setting. Similarly, in Section 
\ref{sec:diagonalization}  the calculation of the logarithmic negativity of a quasi-free state is 
also discussed in reasonable generality. The main results are stated and explained in Section \ref{sec:HOL}.

%%%%%%%%%%%%%%%%%%%%%%%%%
%
%       Set-Up + Main Results
%
%
%%%%%%%%%%%%%%%%%%%%%%%%%%%%%

\section*{Acknowledgements}

This collaboration was made possible by a Summer 2012 {\em Research in Teams} project at 
the Erwin Schr\"o\-dinger International Institute for Mathematical Physics, University of Vienna, 
support from which we gratefully acknowledge. 

\section{Harmonic Oscillator Lattices} \label{sec:HOL}

Our approach to proving entanglement area laws consists of two steps.
The first reduces estimating an appropriate measure of entanglement for the
oscillator lattice systems to specific properties of certain one-particle operators,
known in the Anderson localization literature as {\em eigenfunction correlators}.
We carry out the first step (Theorems~\ref{thm:gsal} and \ref{thm:tsal})  in a general
context which we explain in Section~\ref{subsec:setup}. The second step requires a detailed analysis 
of the localization properties near the bottom of the spectrum of a one-particle system, 
which at present has been completed only in a number of specific examples, including 
the standard Anderson models. See Section~\ref{subsec:apps} for a detailed description of these 
examples.

\subsection{Set-Up and Main Results} \label{subsec:setup} Consider a graph 
$G = (\Gamma, \mathcal{E})$ where $\Gamma$ is a countable set of vertices, often called
sites, and $\mathcal{E}$ is a set of undirected edges, i.e., pairs of vertices. 
We will assume that $G$ is connected, i.e., for any $x,y \in \Gamma$ with $x \neq y$, there exists a path $\gamma_{x,y} = 
(x=x_0, x_1, \cdots, x_n = y)$ of finitely many vertices with $(x_{j-1},x_j) \in \mathcal{E}$ for all 
$1 \leq j \leq n$. 
By $d(x,y)$ we will denote the distance between $x$ and $y$, defined as the minimum number of 
edges in a path connecting $x$ and $y$ and set $d(x,x) = 0$.
In addition, we will assume that the graph is of bounded degree,  i.e.\
\begin{equation} \label{boundeddegree}
N_{\rm max} := \sup_{x\in \Gamma} |\{ y\in \Gamma: (x,y) \in {\mathcal E}\}| < \infty.
\end{equation}
This assumption is equivalent to the existence of a constant $\mu>0$ such that
\begin{equation}
C_{\mu} = \sup_{x \in \Gamma} \sum_{y \in \Gamma} e^{- \mu d(x,y)} < \infty.
\label{Cmu}\end{equation}
If the degree is bounded by $N_{\rm max}$, it suffices to take $\mu > \log N_{\rm max}$
for \eq{Cmu} to hold. Conversely, if \eq{Cmu} holds for some $\mu>0$, one has 
$N_{\rm max} \le C_{\mu} e^{\mu}$. 

For many interesting examples, $G=(\mathbb{Z}^{\nu}, \mathcal{E})$ for
some integer $\nu \geq1$ and edge set given by nearest neighbor pairs, e.g., with respect to the $\ell^1$-metric,
but this is just a special case. We also have results for graphs with exponentially growing volume such as the Bethe lattice.

To formulate area laws, we need a notion of boundary. For this let
$\Lambda_0 \subset \Lambda \subset \Gamma$ and assume that $\Lambda_0$ is finite; $\Lambda$ may be infinite.
By $\partial \Lambda_0$, we will denote the boundary of $\Lambda_0$ which is given by
\begin{equation} \label{boundary}
\partial\Lambda_0= \{x\in\Lambda_0 \mid \mbox{there exists } y\in \Lambda \setminus \Lambda_0 \mbox{ with }
(x,y) \in \mathcal{E} \}.
\end{equation}
Although $\partial\Lambda_0$ depends on $\Lambda$, in general this dependence is of little importance
and we suppress it in the notation. 

Given a graph $G$ as above, we will consider oscillator systems defined as follows.
To each site $x \in \Gamma$, we associate a Hilbert space $L^2( \mathbb{R}, dq_x)$
where we have used $q_x$ to denote the spatial variable. It will be clear that our methods 
easily extend to the case where the single site Hilbert space is taken to be 
$L^2(\mathbb{R}^n, dq_x)$, and so we restrict our
attention to the case of one-dimensional oscillators; mainly to ease the notation. For any finite 
set $\Lambda \subset \Gamma$, a Hilbert space $\mathcal{H}_{\Lambda}$ is defined by setting
\begin{equation} \label{eq:L2spacedef}
\mathcal{H}_{\Lambda} = \bigotimes_{x \in \Lambda} L^2( \mathbb{R}, dq_x) 
= L^2( \mathbb{R}^{\Lambda}, dq)
\end{equation}
with $q = (q_x)_{x \in \Lambda}$. In each finite volume 
$\Lambda$ and for any $x \in \Lambda$, we will also use the notation $q_x$ to denote the 
position operator, i.e. the operator of multiplication by $q_x$ in $\mathcal{H}_{\Lambda}$,
and by $p_x = -i \partial / \partial q_x$ we denote the corresponding momentum operator. By 
standard results, see e.g. \cite{ReedSimon2}, these operators are, on suitable domains, self-adjoint and satisfy the commutation relations
\begin{equation} \label{eq:CCR}
[q_x, q_y] = [p_x, p_y] = 0 \quad \mbox{and} \quad [q_x, p_y] = i \delta_{x,y} \idty 
\quad \mbox{for all } x, y \in \Lambda \, .
\end{equation}

The models we consider will be defined in terms of two real-valued sequences 
$\{ h^{(q)}_{x,y} \}_{x,y \in \Gamma}$ and
$\{ h^{(p)}_{x,y} \}_{x,y \in \Gamma}$. For any finite $\Lambda \subset \Gamma$, we will denote by
$h^{(q)}_{\Lambda}$ the $| \Lambda| \times |\Lambda|$ matrix with entries 
$(h^{(q)}_{\Lambda})_{x,y} = h^{(q)}_{x,y}$ for
all $x,y \in \Lambda$ and similarly $h^{(p)}_{\Lambda}$. Throughout this work, we will assume 
the following.

\begin{assumption}\label{ass:bobs}
There exists a non-decreasing, exhaustive sequence of finite volumes $\Lambda_n \subset 
\Gamma$, for $n \geq 1$, on which the matrices $h^{(q)}_{\Lambda_n}$ and 
$h^{(p)}_{\Lambda_n}$ are real, symmetric, and positive definite.
Moreover, we further assume that there exists $C< \infty$ for which
\begin{equation} \label{eq:unifnormbounds}
\max \left[ \| h^{(p)}_{\Lambda_n} \|, \| (h^{(p)}_{\Lambda_n})^{-1} \|, \| h^{(q)}_{\Lambda_n} \| 
\right] \leq C
\end{equation}
uniformly in $n$.
\end{assumption}

Note that when $\Gamma$ itself is finite, this assumption may be applied to the constant
sequence $\Lambda_n=\Gamma$. We remark that, while requiring that the $h^{(q)}_{\Lambda_n}$ 
are positive definite and thus invertible, we do not assume a uniform bound on 
$(h^{(q)}_{\Lambda_n})^{-1}$. In our applications this will amount to not requiring a robust ground 
state gap for the oscillator systems to be introduced next.

Now, for any $\Lambda_n$ in a sequence
satisfying Assumption \ref{ass:bobs}, the formula
\begin{eqnarray} \label{eq:localham}
H_n & = & \sum_{x,y \in \Lambda_n} \left( q_x h_{x,y}^{(q)} q_y + p_x h_{x,y}^{(p)} p_y \right)  \\
& = & (q^T, p^T) \mathbb{H}_n (q, p) \nonumber
\end{eqnarray}
defines a self-adjoint operator on $\mathcal{H}_{\Lambda_n}$ which we refer to as a finite 
volume oscillator Hamiltonian. In the final line above, we view $q=(q_x)$ and $p=(p_x)$ as 
column vectors indexed by $x \in \Lambda_n$ with transposes
$q^T$ and $p^T$ regarded as corresponding row vectors, and in this case, this line
is a result of standard matrix multiplication with
\begin{equation} \label{eq:x+ham}
\mathbb{H}_n = \left( \begin{array}{cc} h^{(q)}_{\Lambda_n} & 0 \\ 0 & h^{(p)}_{\Lambda_n} 
\end{array} \right)
\end{equation}
When $\Lambda_n$ is understood to be fixed, we will just write $h^{(q)}$ and $h^{(p)}$ to 
ease notation. More general Hamiltonians could also be considered for the conditional 
statements made below. However, the only examples for which we can verify that the 
conditional statements hold are of the form described above, and so we will restrict our 
attention to this case.

It is well-known (e.g.\ \cite{NSS} for more details) that $H_n$ can be written as a system of free Bosons, i.e.,
\begin{equation} \label{eq:diagham}
H_n = \sum_{\ell = 1}^{|\Lambda_n|} \gamma_{\ell} \left( 2 b_{\ell}^* b_{\ell} + \idty \right)\,,
\end{equation}
where the operators $b_{\ell}$ satisfy canonical commutation relations, i.e.,
\begin{equation}
[b_{\ell}, b_{\ell'}] = [b_{\ell}^*, b_{\ell'}^*] = 0 \quad \mbox{and} \quad [b_{\ell}, b_{\ell'}^*] 
= \delta_{\ell , \ell'} \idty \quad \mbox{for all } \ell, \ell' \in \{ 1, 2, \cdots, |\Lambda_n|\} ,
\end{equation}
and the numbers $\gamma_{\ell}>0$ are the square roots of the eigenvalues of the positive 
definite matrix
\begin{equation} \label{effsph}
h_n = (h^{(p)})^{1/2} h^{(q)} (h^{(p)})^{1/2} \, .
\end{equation}
For this reason, we often refer to $h_n^{1/2}$ as the effective single-particle Hamiltonian 
corresponding to $H_n$, regarding it as a self-adjoint operator on $\ell^2(\Lambda_n)$.

In this work, we are interested in random models. In this case, we regard the components of
sequences defining $h^{(q)}$ and $h^{(p)}$ as random variables on a probability space 
$(\Omega, \mathbb{P})$. We will assume that there is a deterministic sequence of volumes 
$\Lambda_n$ for which Assumption \ref{ass:bobs} holds almost surely. By $\mathbb{E}(X)$, 
we will denote the expectation (average) of a random
variable $X$ on $\Omega$ with respect to $\mathbb{P}$.

Our first result concerns the ground state of the random Hamiltonian $H_n$.
{F}rom the form of (\ref{eq:diagham}), it follows that, almost surely, $H_n$ has a unique, 
normalized ground state $\Omega_n \in \mathcal{H}_{\Lambda_n}$ which is characterized 
by $b_{\ell} \Omega_n = 0$ for all $1 \leq \ell \leq |\Lambda_n|$.
Let us denote by $\rho_n$ the orthogonal projection onto $\Omega_n$.

We study the bipartite entanglement in the ground state (and later also
in the equilibrium states) with respect to a partition of the system.
Explicitly, fix a finite set  $\Lambda_0 \subset \Gamma$ and 
for any $n\geq 1$  large enough so that $\Lambda_0\subset \Lambda_n$, 
write 
$\mathcal{H}_{\Lambda_n} = \mathcal{H}_1 \otimes \mathcal{H}_2$ where
\begin{equation} \label{decomp}
\mathcal{H}_1 = \bigotimes_{x \in \Lambda_0}L^2( \mathbb{R}, dq_x) 
\quad \mbox{and} \quad \mathcal{H}_2 
= \bigotimes_{x \in \Lambda_n \setminus \Lambda_0}L^2( \mathbb{R}, dq_x)
\end{equation}
Denote by $\rho_n^1 = \Tr_{\mathcal{H}_2}\rho_n$ the reduction of the ground state 
projector to $\mathcal{H}_1$, and for
any non-negative $\rho$ with trace 1, let $S(\rho) = - \Tr \rho \ln\rho$ be the von Neumann 
entropy of $\rho$.

\begin{thm} \label{thm:gsal} 
Under the conditions stated above, in particular \eq{Cmu} and
Assumption \ref{ass:bobs}, and assuming further that there is a 
$C'< \infty$ and a $\mu' \geq \mu$ for which
\begin{equation} \label{eq:gsloccond}
\mathbb{E} \left( \left| \left\langle (h^{(p)})^{1/2} \delta_x, h_n^{-1/2} (h^{(p)})^{1/2} \delta_y 
\right\rangle \right| \right) \leq C' e^{- \mu' d(x,y)}
\end{equation}
for all $n \geq 1$ and all $x,y \in \Lambda_n$. Then there exists $C''< \infty$ for which
\begin{equation} 
\mathbb{E} \left( S(\rho_n^1) \right) \leq C'' | \partial \Lambda_0|
\label{gsal}\end{equation}
for all $n\ge 1$.
\end{thm}

It has been argued that a bound on the entropy of entanglement of the type \eq{gsal}
indicates that the ground state properties are computable. See the discussion in 
Section VI of  \cite{Eisert2010} for an overview and \cite{Evenbly2010}
for a more nuanced discussion of the question in the case of oscillator lattices.

Our proof of Theorem~\ref{thm:gsal} uses the following bound for the entropy of the
restriction of a pure state to one factor of a bipartite decomposition
$\mathcal{H} = \mathcal{H}_1 \otimes \mathcal{H}_2$:
\begin{equation} \label{eq:logneg}
S(\Tr_{\mathcal{H}_2}\rho)\leq\log \Vert \rho^{T_1}\Vert_1=\mathcal{N}(\rho).
\end{equation}
Here, the RHS is the log of the 1-norm of the partial transpose of the pure state $\rho$
with respect to the given decomposition. This quantity is called  the {\it logarithmic negativity}
of the density matrix $\rho$ with respect to the decomposition and is denoted by
$\mathcal{N}(\rho)$. See the Appendix for a detailed discussion of partial transposes and, 
in particular, Lemma~\ref{lem:purebd} for a proof of \eq{eq:logneg}.

For any $\beta>0$, the equilibrium state (aka thermal state) at inverse temperature $\beta$ 
of the finite system with Hamiltonian $H_n$ is given by the density matrix
\begin{equation} \label{thermalstate}
\rho_{\beta,n}=\frac{e^{-\beta H_n}}{\Tr e^{-\beta H_n}} .
\end{equation}
As discussed in \cite{Vidal2002} the logarithmic negativity is a reasonable measure of 
the bipartite entanglement not only for pure states but also for mixed states (such as
thermal states). In particular it is an {\em entanglement monotone}, see e.g. \cite{Vidal2002} for details.

\begin{thm} \label{thm:tsal} Fix $\beta >0$. Under the conditions stated above,
assume further that there is a $C'< \infty$ and a $\mu' \geq \mu$ for which
\begin{equation} \label{eq:tsloccond}
\mathbb{E} \left( \left| \left\langle (h^{(p)})^{1/2} \delta_x, h_n^{-1/2} \tanh ( \beta (h_n)^{1/2} ) (h^{(p)})^{1/2} \delta_y \right\rangle \right| \right) \leq C' e^{- \mu' d(x,y)}
\end{equation}
for all $n \geq 1$ and all $x,y \in \Lambda_n$. Then there exists $C''< \infty$ for which
\begin{equation} \label{tsal}
\mathbb{E} \left( \mathcal{N}(\rho_{\beta, n}) \right) \leq C'' | \partial \Lambda_0|
\end{equation}
for all $n\ge 1$, where $\rho_{\beta,n}$ is the density matrix defined in \eq{thermalstate} and
$\mathcal{N}(\rho_{\beta, n})$ is the logarithmic negativity with respect to the decomposition defined in (\ref{decomp}).
\end{thm}

The proof of these results is given in Section \ref{sec:proof}, which in turn relies on results we
derive in Section \ref{sec:diagonalization} and the Appendix.

\subsection{Applications} \label{subsec:apps}

Applications of Theorems~\ref{thm:gsal} and \ref{thm:tsal} consist in verifying the localization bounds (\ref{eq:gsloccond}) and (\ref{eq:tsloccond}) for the underlying effective single-particle Hamiltonian $h_n$ in concrete special cases of oscillator systems $H_n$. 

Among available results in single-particle localization theory, bounds of the form (\ref{eq:gsloccond}) or (\ref{eq:tsloccond}) correspond to strong forms of localization, which have only been rigorously established for the Anderson models (or models closely related to it). This leads us to consider the special case
\begin{equation} \label{eq:localham1}
H_n = \sum_{x\in \Lambda_n} \left( \frac{1}{2m} p_x^2 + \frac{g k_x}{2} q_x^2 \right) + \sum_{(x,y)\in {\mathcal E}_n} \lambda (q_x-q_y)^2
\end{equation}
of (\ref{eq:localham}). Here ${\mathcal E}_n = \{(x,y) \in {\mathcal E}: x,y\in \Lambda_n\}$. The masses $m$, coupling parameters $\lambda$  and disorder parameter $g$ are positive constants. 

\begin{assumption} \label{ass:springconst}
The $\{k_x\}_{x\in \Gamma}$ are independent, identically distributed random variables, which are absolutely continuous with bounded density $\tilde{\rho}$ supported in $[0,k_{\rm max}]$ for some $k_{\rm max}>0$.
\end{assumption}

We thus have $h^{(p)} = \frac{1}{2m}\idty$ and $h^{(q)}$ becomes the Anderson model characterized by its quadratic form
\begin{equation} \label{eq:Anderson}
\langle f, h^{(q)}g \rangle = \sum_{(x,y) \in {\mathcal E}} \lambda \overline{(f(y)-f(x))} (g(y)-g(x)) + \frac{g}{2} \sum_{x\in \Gamma} k_x \overline{f(x)} g(x),
\end{equation}
$f$, $g \in \ell^2(\Gamma)$. Their restrictions to $\Lambda_n$ satisfy (\ref{eq:unifnormbounds}) due to the boundedness of the $k_x$ and the fact that $\Gamma$ is of bounded degree, in fact,
\begin{equation} \label{eq:hqnormbound}
\|h^{(q)}_{\Lambda_n}\| \le 2\lambda N_{\rm max} + \frac{1}{2} gk_{\rm max}.
\end{equation}

Moreover, $h_n = \frac{1}{2m} h^{(q)}_{\Lambda_n}$ and the bounds (\ref{eq:gsloccond}) and (\ref{eq:tsloccond}) become equivalent to the existence of $C'<\infty$ and $\mu'\ge \mu$ such that, for all $n\ge 1$ and $x,y \in \Lambda_n$,
\begin{equation} \label{eq:gsloc1}
\E \left( |\langle \delta_x, h_n^{-1/2} \delta_y \rangle| \right) \le C' e^{-\mu' d(x,y)}
\end{equation}
and
\begin{equation} \label{eq:tsloc1}
\E \left( |\langle \delta_x, h_n^{-1/2} \tanh(\beta h_n^{1/2}) \delta_y \rangle| \right) \le C' e^{-\mu' d(x,y)},
\end{equation}
respectively.

Due to Assumption~\ref{ass:springconst} it holds almost surely that $k_x>0$ for all $x\in \Gamma$. Thus, almost surely, the matrices $h_n$ are positive definite in any finite volume $\Lambda_n$. However, as the support of the density $\tilde{\rho}$ contains $0$, there is no volume independent deterministic lower bound of the form $h_n \ge C>0$. A lower bound of this form is essentially what was used in \cite{Plenio2005, Cramer2006} to show that deterministic exponential decay bounds as in (\ref{eq:gsloc1}) and (\ref{eq:tsloc1}) hold, thus implying area laws as in Theorems~\ref{thm:gsal} and \ref{thm:tsal}.

In terms of the oscillator system Hamiltonian $H_n$ this means that an area law for the entanglement entropy of ground and thermal states was found to be a consequence of a robust ground state gap. This is seen by representing $H_n$ as a free Boson system (\ref{eq:diagham}), which shows that the ground state gap of $H_n$ is given by
\begin{equation}
2 \min_{\ell} \gamma_{\ell} = 2 \min \sigma(h_n^{1/2}) = 2 \left( \min \sigma(h_n) \right)^{1/2},
\end{equation}
see \cite{NSS} for more details.

A central goal of our work here is to show that in {\em disordered} oscillator systems it is not necessary for an area law to require a robust ground state gap, as long as averages are considered on the left hand sides of (\ref{eq:gsloc1}) and (\ref{eq:tsloc1}). Using a term which was first proposed in a related context for disordered quantum spin systems in \cite{Hastings2010} (see also \cite{HamzaSimsStolz}), we can argue that the localization properties of the single-particle operator $h_n$ lead to a {\em mobility gap}, which has consequences for the disordered many-body system $H_n$ similar to those of a robust ground state gap for a deterministic system.

Single-particle localization bounds similar to (\ref{eq:gsloc1}) and (\ref{eq:tsloc1}) have been used in \cite{NSS} to prove certain characteristics of many-body localization in disordered oscillator systems, such as dynamical localization in the form of zero-velocity Lieb-Robinson bounds, and exponential decay of ground state and thermal state correlations. Of particular interest in this context is the left hand side of (\ref{eq:gsloc1}), which, due to the absence of a robust ground state gap, presents an example of a {\em singular eigenfunction correlator}. Appendix A of \cite{NSS} provides a detailed discussion of localization bounds for singular eigenfunction correlators, based on earlier results in the theory of Anderson localization (such as recently reviewed in \cite{Stolz2011}). In particular, this leads quite directly to the first of the following applications of our results.

\begin{thm}[Lattice Systems] \label{thm:latticesystems}
Let $G= (\Z^{\nu}, {\mathcal E})$ with edge set given by nearest neighbor pairs and $\Lambda_n = [-n,n]^{\nu} \cap \Z^{\nu}$, and let $H_n$ be the disordered oscillator system (\ref{eq:localham1}) over $G$ at fixed disorder $g=1$, satisfying Assumption~\ref{ass:springconst}. Then the ground state and thermal states of $H_n$ satisfy area laws of the form (\ref{gsal}) and (\ref{tsal}), respectively. 
\end{thm}

\begin{proof}
For $G= \Z^{\nu}$ the volumes $|\{y\in G: |x-y|_1 \le n\}|$ grow polynomially in $n$. Thus (\ref{Cmu}) holds for any $\mu>0$ and, due to Theorems~\ref{thm:gsal} and \ref{thm:tsal}, it suffices to verify (\ref{eq:gsloc1}) and (\ref{eq:tsloc1}) for some $\mu'>0$. Both bounds follow as special cases from Proposition~A.3(c) in \cite{NSS}, as the functions $\varphi_1(t) = t^{-1/2}$ and $\varphi_2(t) = t^{-1/2} \tanh(\beta t^{1/2})$ both have analytic extensions to the half plane $\{z: \mbox{Re}\,z>0\}$ and for $t\in (0,\infty)$ satisfy bounds of the form $|\varphi(t)|\le Ct^{\alpha}$ for some $\alpha>-1$. In fact, for $\varphi_2$ one can work with $\alpha=0$ as $\tanh(\beta t^{1/2}) \sim t^{1/2}$ near $0$.
\end{proof}

Without going into detail, we mention two natural directions in which Theorem~\ref{thm:latticesystems} can be generalized:

\vspace{.5cm}

(a) We have considered $h^{(p)} = c\idty$, but this can be generalized to larger classes of matrices $h^{(p)}$, at least for the ground state case. Using (\ref{effsph}) we see that the left hand side of (\ref{eq:gsloccond}) has the form 
\begin{equation} 
\E(|\langle \delta_x, (h^{(p)})^{1/4} (h^{(q)})^{-1/2} (h^{(p)})^{1/4} \delta_y \rangle|).
\end{equation}
One can prove exponential decay for this using the same arguments as in the proof of Theorem~\ref{thm:latticesystems} as long as one has an a-priori exponential decay bound for the matrix elements of $(h^{(p)})^{1/4}$. This can be shown if the $h^{(p)}$ are positive definite, diagonally dominant band matrices, satisfying the uniform norm bounds required in (\ref{eq:unifnormbounds}). In this case one can use the analyticity of $\varphi(x)=x^{1/4}$ in the right complex half plane to show exponential decay of matrix elements of $(h^{(p)})^{1/4}$, using arguments similar to those in, e.g., \cite{CramerEisert2006}.  For example, one can choose $h^{(p)} = c\idty + \delta T$, where $T$ is the next-neighbor hopping operator and $\delta < c/(2\nu)$.

Dealing with thermal states would require more work, as $(h^{(p)})^{1/2} h_n^{-1/2} \tanh ( \beta (h_n)^{1/2} ) (h^{(p)})^{1/2}$, appearing in $(\ref{eq:tsloccond})$, does not factorize into fractional powers of $h^{(p)}$ and $h^{(q)}$.

\vspace{.5cm}

(b) It is also natural to ask if Theorem~\ref{thm:latticesystems} extends to general graphs $G=(\Gamma, {\mathcal E})$ as long as they have polynomially bounded volume growth, i.e.\ the sets $\Lambda_n(x) = \{y: d(x,y)\le n\}$ grow polynomially in $n$ (uniform in $x$). The crucial ingredient into the proof of Proposition~A.3(c) of \cite{NSS}, which we use above, is the well-known Lifshitz tail argument leading to localization of the single-particle Hamiltonians $h_n$ near $E=0$. This requires to know a deterministic lower bound of the form $E_1-E_0 \ge C/n^2$ for the ground state gap of the discrete Graph Laplacian on $\Lambda_n(x)$. We are not aware of a general result establishing such a bound for graphs with polynomial volume growth. But whenever it is known, area laws for the ground and thermal states will follow.

\vspace{.5cm}

In general, graphs of bounded degree have exponential volume growth, with the prototypical example given by the Bethe lattice. It is known that in this case localization proofs for the single-particle Hamiltonians $h_n$ require sufficiently large disorder $g$, see e.g.\ \cite{AizenmanMolchanov} and \cite{AizenmanWarzel}. In fact, the results of \cite{AizenmanWarzel} establish that for the low-disorder Anderson model on the Bethe lattice the extended states regime may extend all the way to the spectral boundaries. However, if the disorder is sufficiently large, then we get area laws on general graphs of bounded degree:

\begin{thm}[Large Disorder] \label{thm:largedisorder}
Let $H_n$ be the disordered oscillator system (\ref{eq:localham1}) on a general graph $G=(\Gamma, {\mathcal E})$, satisfying (\ref{boundeddegree}) and Assumption~\ref{ass:springconst}. If the disorder parameter $g>0$ is sufficiently large, then the ground state and thermal states of $H_n$ satisfy area laws of the form (\ref{gsal}) and (\ref{tsal}), respectively. 
\end{thm}

\begin{proof}
The main difference to the proof of Theorem~\ref{thm:latticesystems} is that we now need to show (\ref{eq:gsloc1}) and (\ref{eq:tsloc1}), respectively,  for some $\mu'>\mu$, with $\mu$ the constant from (\ref{Cmu}). As explained there this means that we need to show that we can choose $\mu' > \log N_{\rm max}$. As a matter of fact, we will show that by increasing $g$ one can choose $\mu'$ {\em arbitrarily} large. This will follow by well established methods, but we will provide some detail as large disorder localization of singular eigenfunction correlators has not previously been discussed in the literature.

For any subset $\Lambda \subset \Gamma$ let $h_{\Lambda}$ be the restriction of the Anderson model (\ref{eq:Anderson}) to $\Lambda$, and let
\begin{equation}
G_{\Lambda}(x,y;z) = \langle \delta_x, (h_{\Lambda}-z)^{-1} \delta_y \rangle
\end{equation}
be its Green function. By adjusting arguments in Section~4 of \cite{Stolz2011}  to the case of general graphs considered here one gets the following fractional moment bound (which is essentially already contained in \cite{AizenmanMolchanov}):

For every $s\in (0,1)$ there exists a constant $C_1 = C_1(s, \tilde{\rho})$ such that
\begin{equation} \label{eq:fmbound}
\E \left( |G_{\Lambda}(x,y;z)|^s \right) \le \frac{C_1}{g^s} e^{-\log(g^s/C_1 N_{\rm max}) d(x,y)}
\end{equation}
for all $\Lambda \subset \Gamma$, $x,y \in \Lambda$, $g>0$ and $z\in \C \setminus \R$. In fact, if $\Lambda$ is finite, then one may choose $z\in \C$, including real values.

Next we relate fractional moments to eigenfunction correlators:

For every $s\in (0,1)$ there exists $C_2 = C_2(s,\tilde{\rho})$ such that
\begin{equation} \label{eq:fmefcor}
\E \left( \sup_{|u|\le 1} |\langle \delta_x, u(h_{\Lambda}) \chi_I(h_{\Lambda}) \delta_y \rangle| \right) \le C_2 \left( |g|^s \int_I \E(|G_{\Lambda}(x,y;E)|^s\,dE \right)^{1/(2-s)}
\end{equation}
for all finite $\Lambda \subset \Gamma$, $x,y\in \Lambda$, $g>0$ and bounded open intervals $I \subset \R$.

This can be proven following the $g$-dependence of the arguments in Section~6 of \cite{Stolz2011}, where the case $g=1$ is considered.

To conclude the proof of (\ref{eq:gsloc1}) and (\ref{eq:tsloc1}) we now apply these bounds to $\Lambda=\Lambda_n$ and, as before, write $h_n = h_{\Lambda_n}$, $n=1,2,\ldots$.

Choosing $u(x)= \beta^{-1} x^{-1/2} \tanh(\beta x^{1/2})$ (which satisfies $|u|\le 1$) and $I=(0,E_{\rm max}) := (0, 2\lambda N_{\rm max} + \frac{1}{2}g k_{\rm max})$ (which by (\ref{eq:hqnormbound}) almost surely contains the entire spectrum of $h_n$) and combining (\ref{eq:fmbound}) and (\ref{eq:fmefcor}) yields
\begin{equation} 
\E \left( |\langle \delta_x, h_n^{-1/2} \tanh(\beta h_n^{1/2}) \delta_y \rangle| \right) \le \beta C_2 (C_1 E_{\rm max})^{\frac{1}{2-s}} e^{-\frac{1}{2-s} \log(g^s/C_1 N_{\rm max}) d(x,y)}.
\end{equation}
This proves (\ref{eq:tsloc1}) for $g$ sufficiently large.

Proving (\ref{eq:gsloc1}) needs a bit more work, as this requires handling a {\it singular} eigenfunction correlator. This can be done by the Riemann sum argument previously used in the proof of Proposition~A.3(b) in \cite{NSS}. Decompose $I=(0,E_{\rm max})$ into
\begin{equation}
I_j = \left( \frac{E_{\rm max}}{j+1}, \frac{E_{\rm max}}{j} \right), \quad j=1,2, \ldots,
\end{equation}
and combine (\ref{eq:fmbound}) and (\ref{eq:fmefcor}) to get
\begin{equation} \label{eq:combinedwisdom}
\E \left( | \langle \delta_x, h_n^{-1/2} \delta_y \rangle| \right) \le C_2 C_1^{\frac{1}{2-s}} \left( \sum_{j=1}^{\infty} \left(\frac{E_{\rm max}}{j+1}\right)^{-1/2} |I_j|^{\frac{1}{2-s}} \right) e^{-\frac{1}{2-s} \log(g^s/C_1 N_{\rm max}) d(x,y)}.
\end{equation}
For $s\in (2/3,1)$ the series in (\ref{eq:combinedwisdom}) is summable, which concludes the proof of (\ref{eq:gsloc1}), again for $g$ sufficiently large.
\end{proof}

We conclude our discussion of applications by acknowledging that the types of disorder in oscillator systems which we have been able to handle is rather limited, essentially only covering Anderson-type (diagonal) randomness in $h^{(q)}$. Considering other types of disorder, such as random masses $m$ or coupling constants $\lambda$ in (\ref{eq:localham1}), is physically equally plausible, but not enough is known about the localization properties of the associated single-particle Hamiltonians.  Thus the applications provided here (as well as in \cite{NSS}) motivate further studies of single-particle random Hamiltonians, with the goal of covering other phyically relevant cases.

%%%%%%%%%%%%%%%%%%%%%%%
%
%
%	Quasi-Free Functional and States:
%
%
%%%%%%%%%%%%%%%%%%%%%%%%%

\section{Gaussian States and Their Logarithmic Negativity}
\label{sec:diagonalization}

The main goal of this section is to prove Theorem~\ref{thm:logneg}, which provides a formula
for the logarithmic negativity associated with an arbitrary finite volume $\Lambda_0$ for a class
of quasi-free states, including the ground and thermal states of the harmonic oscillator models
introduced in Section~\ref{sec:HOL}. As was the case in all previous results of this kind
\cite{Audenaert2002,Plenio2005,CramerEisert2006,Cramer2006}, we start from the ideas in \cite{Vidal2002}.

The logarithmic negativity is an upper bound for the entropy of entanglement. In the case of pure 
states ($T=0$), the latter is the von Neumann entropy of the state restricted to the observables 
localized in $\Lambda_0$. The restriction of a quasi-free state is again a quasi-free state and 
this property makes it possible to essentially reduce the calculation to diagonalizing a 
one-particle operator. Calculating the logarithmic negativity means finding the one-norm of
a partial transpose of the density matrix of the state or, equivalently, finding the norm of the
partial transpose of the state regarded as a linear functional on the algebra of observables.
The partial transpose of a quasi-free state, although in general not a state, is again a quasi-free 
(i.e.\ Gaussian) functional. This property makes is possible to find an explicit formula for the 
logarithmic negativity, see (\ref{lognegprop}). In Section~\ref{sec:proof} we prove the 
Area Law bound based on this formula, that is we prove Theorems \ref{thm:gsal} and \ref{thm:tsal}.
Since quasi-free functionals other than states, in particular, quasi-free functionals that  are not 
positive, have not been widely studied, we provide in this section the necessary elements needed
for the proof of  Theorem~\ref{thm:logneg} in reasonable detail.

The strategy of this section is as follows. Fix $\Lambda_0 \subset \Gamma$ finite. 
As is discussed in Section~\ref{sec:HOL}, for any finite $\Lambda \subset \Gamma$,
both the ground and thermal states of the oscillator systems we consider can be expressed
in terms of a density matrix $\rho$:
\begin{equation}
\omega(A) = {\rm Tr} \left[ \rho A \right] \quad \mbox{for any } A \in B(\mathcal{H}_{\Lambda}) \, .
\label{densitymatrix}\end{equation}
Here we have suppressed the dependence of the
state $\omega$ and the density matrix $\rho$ on the finite volume $\Lambda$. For $\Lambda$ with $\Lambda_0 \subset \Lambda$,
the logarithmic negativity of $\rho$ is
defined by 
\begin{equation} \label{eq:logneg3}
\mathcal{N}( \rho) = \log \left( \| \rho^{T_1} \|_1 \right) 
\end{equation}
where $\rho^{T_1}$ is the partial transpose of $\rho$ with respect to the decomposition 
as in (\ref{decomp}).  The one-norm of $\rho^{T_1}$ equals the norm of the linear functional
\begin{equation}
\omega^{T_1}(A) = {\rm Tr} \left[ \rho^{T_1} A \right] \quad \mbox{for any } 
A \in B(\mathcal{H}_{\Lambda}) \, ,
\end{equation}
which is well-defined exactly when $\Vert \rho^{T_1}\Vert_1 <\infty$ (proving the latter will be part of our argument below).
Motivated by this relationship, we will start by studying the partial transpose of quasi-free
functionals on the Weyl algebra (see Sections~\ref{sec:set-up} and \ref{sec:basex}),
defined by
\begin{equation}
\omega^{T_1}(W(f)) = e^{- \frac{1}{4} (f, \tilde{M} f)}
\end{equation}
in terms of a real, symmetric, positive definite matrix $\tilde{M}$. This functional need 
not be a state (it is not necessarily positive), but a version of Williamson's Theorem (see
Proposition~\ref{prop:williamson} below), implies that there exists a symplectic matrix 
$S$ which diagonalizes  $\tilde{M}$ and therefore
\begin{equation}
\omega^{T_1}(W(Sf)) = \prod_j e^{- \frac{1}{4} \lambda_j |f(j)|^2}
\end{equation} 
where $\lambda_j>0$ are the symplectic eigenvalues of $\tilde{M}$; again, more on this can be 
found in Section~\ref{sec:will}. By explicit construction,
we demonstrate in Section~\ref{sec:1norm} the existence of a trace class operator $\tilde{\rho}$
such that
\begin{equation}
\omega^{T_1}(W(Sf)) = {\rm Tr} \left[ \tilde{\rho} W(f) \right] \, .
\end{equation}
Since there can only be one trace class operator satisfying this relationship 
(see Lemma~\ref{lem:weakdense} below), we conclude that 
$\tilde{\rho}$ is unitarily equivalent to $\rho^{T_1}$. 
It is then straightforward to find the one-norm of $\tilde{\rho}$ using its explicit form, and
this yields the expression of the logarithmic negativity in Theorem~\ref{thm:logneg}.
%%%%%%%%%%%%%%%%%%%%%%
%
%
%          On CCR algebras and quasi-free functionals
%
%
%%%%%%%%%%%%%%%%%%

\subsection{On Weyl algebras and quasi-free functionals} \label{sec:set-up}

We begin by introducing Weyl algebras, or CCR algebras, in the abstract setting. 
Here, we are brief and refer the interested reader to \cite{BratRob} for further background and details.
Next, we describe quasi-free functionals on the Weyl algebra; these can be regarded as generalizations of the 
well-studied class of quasi-free states. The need for this generalization stems from our interest in logarithmic negativity. In particular,
the partial transpose of a density matrix associated to a state on the Weyl algebra induces a functional that is
not necessarily a state.  

Let $\mathcal{D}$ be any real-linear space equipped with a non-degenerate,
symplectic bilinear form $\sigma$, i.e.
$\sigma : \mathcal{D} \times \mathcal{D} \to \mathbb{R}$ with the property that
if $\sigma(f,g) = 0$ for all $f \in \mathcal{D}$, then $g = 0$, and
\begin{equation} \label{eq:symp}
\sigma(f,g) = - \sigma(g,f) \quad \mbox{for all } f, g \in \mathcal{D} .
\end{equation}
The Weyl operators over $\mathcal{D}$ are introduced by associating non-zero elements
$W(f)$ to each $f \in \mathcal{D}$ which satisfy
\begin{equation} \label{eq:invo}
W(f)^* = W(-f) \quad \mbox{for each } f \in \mathcal{D} \, ,
\end{equation}
and
\begin{equation} \label{eq:weylrel}
W(f) W(g) = e^{-i \sigma(f,g)/2} W(f+g) \quad \mbox{for all } f, g \in \mathcal{D} \, .
\end{equation}
As is proven e.g. in Theorem 5.2.8 \cite{BratRob}, there is a unique, up to $*$-isomorphism, $C^*$-algebra generated
by these Weyl operators with the property that $W(0) = \idty$, $W(f)$ is unitary
for all $f \in \mathcal{D}$, and $\| W(f) - \idty \| = 2$ for all $ f \in \mathcal{D} \setminus \{0 \}$.
This algebra, commonly known as the Weyl algebra (also CCR algebra) over
$\mathcal{D}$ will be denoted by $\mathcal{W}( \mathcal{D})$.

Fix a real-linear space $\mathcal{D}$ and the corresponding Weyl algebra 
$\mathcal{W} = \mathcal{W}(\mathcal{D})$.
$\omega$ is said to be a quasi-free functional on $\mathcal{W}$ if
\begin{equation} \label{quasi-free}
\omega(W(f)) = e^{i r(f) - \frac{1}{4}s(f,f)} \quad \mbox{for all } f \in \mathcal{D}
\end{equation}
where $r$ is a real-linear functional and $s$ is a symmetric, real bilinear form on $\mathcal{D}$.

It is clear that equation (\ref{quasi-free}) uniquely defines a linear functional on
a dense subalgebra of $\mathcal{W}$. Due 
to the form of (\ref{quasi-free}), such functionals are also referred to as Gaussian.
Not all these functionals are states, i.e. positive linear functionals on $\mathcal{W}$,
even if we assume they are continuous. 
In fact, it is well-known, see \cite{Ver}, that a functional of the form (\ref{quasi-free}) is a 
state if and only if
\begin{equation} \label{positivity}
\sigma(f,g)^2\leq s(f,f) s(g,g) \quad \mbox{for all } f,g \in \mathcal{D} \, .
\end{equation}
Observe that the inequality (\ref{positivity}) above is equivalent to
\begin{equation}
2\sigma(f,g)\leq s(f,f) + s(g,g)\,.
\end{equation}

%%%%%%%%%%%%%%%%%%%%%%
%
%
%          Basic Examples
%
%
%%%%%%%%%%%%%%%%%%

\subsection{Oscillator Model Examples} \label{sec:basex}

In this section, we briefly review the fact that the ground and thermal states of the harmonic oscillator
lattices introduced in Section~\ref{subsec:setup} are quasi-free, in the sense discussed above. 
We end this section with a discussion on corresponding partially transposed functionals. 

Let $\Lambda \subset \Gamma$ be finite. The role of $\mathcal{D}$ is played by the
complex Hilbert space $\mathcal{D}_{\Lambda} = \ell^2(\Lambda)$ with the
symplectic form related to the inner product by
\begin{equation}
\sigma(f,g) = \mbox{Im} \left[ \langle f, g \rangle \right] \, .
\end{equation}
The corresponding Weyl algebra $\mathcal{W}_{\Lambda}$ has a concrete realization: for each $f \in \ell^2(\Lambda)$, it is well-known that
\begin{equation} \label{eq:concreteWeyl}
W(f) = {\rm exp}\left[ i \sum_{j \in \Lambda} \left( {\rm Re}[f(j)] q_j + {\rm Im}[f(j)] p_j \right) \right]
\end{equation}
defines a unitary Weyl operator in $\mathcal{B}( \mathcal{H}_{\Lambda})$ satisfying
(\ref{eq:invo}) and (\ref{eq:weylrel}) above. Here, for each $j \in \Lambda$,
$q_j$ and $p_j$ are the position and momentum operators introduced in Section~\ref{subsec:setup}.

The following basic fact is important for us. Here $B_1(\mathcal{H}_{\Lambda})$ denotes the trace class operators on $\mathcal{H}_{\Lambda}$.

\begin{lemma} \label{lem:weakdense}
If $A \in B_1(\mathcal{H}_{\Lambda})$ and
\begin{equation} \label{eq:weakzero}
{\rm Tr}\,[ A\,W(f)] = 0 \quad \mbox{for all $f \in \ell^2(\Lambda)$},
\end{equation}
then $A=0$.
\end{lemma}

\begin{proof}
This is a consequence of irreducibility of the Weyl algebra $\mathcal{W}_{\Lambda}$ in $B(\mathcal{H}_{\Lambda})$ (e.g.\ Proposition~5.2.4(3) of \cite{BratRob}, the Fock space representation of the Weyl operators used there is equivalent to the representation (\ref{eq:concreteWeyl}) when working in the Hermite function basis of $\mathcal{H}_{\Lambda} = L^2(\mathbb{R}^{\Lambda})$)  and von Neumann's double-commutant Theorem, showing that $\mathcal{W}_{\Lambda}$ is weakly dense in $B(\mathcal{H}_{\Lambda})$. The identity (\ref{eq:weakzero}) implies that ${\rm Tr}\,AC=0$ for all $C\in \mathcal{W}_{\Lambda}$. This carries over to general $C\in B(\mathcal{H}_{\Lambda})$ due to the fact that $C_n \stackrel{w}{\to} C$ implies ${\rm Tr}\,AC_n \to {\rm Tr}\,AC$. Finally, use that 
\begin{equation}
\|A\|_1 = \sup_{C\in B(\mathcal{H}_{\Lambda}), \,\|C\|\le 1} |{\rm Tr}\,AC|
\end{equation}
to conclude that $A=0$.
\end{proof}

Given an operator $H_{\Lambda}$, as in (\ref{eq:localham}), denote by $\rho_{\Lambda}$
the orthogonal projection onto the unique, normalized ground state of $H_{\Lambda}$.
A ground state functional $\omega_{\Lambda}$ on $\mathcal{W}_{\Lambda}$ is defined by setting
\begin{equation}
\omega_{\Lambda}(W(f)) = {\rm Tr} \left[ \rho_{\Lambda} W(f) \right]  \quad \mbox{for all } f \in \ell^2( \Lambda) \, .
\end{equation}
Here, and in what follows, we will regard the set $\Lambda$ as fixed and simply write $\omega$ and $\rho$. 
It will also be convenient to identify $\ell^2(\Lambda) = \ell^2(\Lambda; \mathbb{C})$ with
$\ell^2(\Lambda; \mathbb{R}) \oplus \ell^2(\Lambda; \mathbb{R})$, i.e.,
\begin{equation} \label{real_basis}
f \in \ell^2(\Lambda; \mathbb{C}) \quad \sim \quad \tilde{f} = \left( \begin{array}{c} {\rm Re}[f] \\ {\rm Im}[f] \end{array} \right) \in \ell^2(\Lambda; \mathbb{R}) \oplus \ell^2(\Lambda; \mathbb{R}) \, .
\end{equation}
In this case, one calculates that
\begin{equation} \label{defJ}
\sigma(f,g) = (J \tilde{f}, \tilde{g}) \quad \mbox{where} \quad J = \left( \begin{array}{cc} 0 & - \idty \\ \idty & 0 \end{array} \right) \, ,
\end{equation}
and $( \cdot, \cdot)$ is the inner product on the direct sum. To ease the notation, we will just write $f = \tilde{f}$ where this 
identification is to be understood.
A well-known calculation, see e.g. \cite{Schuch2006}, shows that
\begin{equation} \label{eq:Weylexp}
\omega(W(f)) = e^{- \frac{1}{4} ( f, M f )}
\end{equation}
where $M$ is the positive definite matrix
\begin{equation} \label{Mmat}
M = \left( \begin{array}{cc} (h_{\Lambda}^{(p)})^{1/2} h_{\Lambda}^{-1/2} (h_{\Lambda}^{(p)})^{1/2} & 0 \\ 0 & (h_{\Lambda}^{(p)})^{-1/2} h_{\Lambda}^{1/2} (h_{\Lambda}^{(p)})^{-1/2} \end{array} \right)
\end{equation}
and $h_{\Lambda} = (h_{\Lambda}^{(p)})^{1/2} h_{\Lambda}^{(q)} (h_{\Lambda}^{(p)})^{1/2}$ is as in (\ref{effsph}). $M$, as above, is proportional to the real part of the ground state covariance
matrix, i.e., the $2|\Lambda| \times 2 | \Lambda|$ matrix $C$ given by
\begin{equation}
C = (c_{jk}) \quad \mbox{where} \quad c_{jk} = {\rm Tr} \left[ \rho \, x_j x_k \right]  \quad \mbox{and} \quad x = \left( \begin{array}{c} q \\ p \end{array} \right) \, .
\end{equation}
In fact, one easily checks that $2C = M - iJ$ with $J$ as in (\ref{defJ}).

Similarly, for any $\beta>0$ and $\Lambda$ finite, a thermal state functional $\omega_{\beta}$ on
$\mathcal{W}_{\Lambda}$ is given by
\begin{equation}
\omega_{\beta}(W(f)) = {\rm Tr}[ \rho_{\beta} W(f)]  = e^{- \frac{1}{4} ( f, M_{\beta} f)}
\end{equation}
where $\rho_{\beta}$ is the thermal state density matrix, see e.g. (\ref{thermalstate}), and the final equality above 
is again the result of a well-known calculation. Here
\begin{equation} \label{eq:mbeta}
M_{\beta} = \left( \begin{array}{cc} (h_{\Lambda}^{(p)})^{1/2} \coth( \beta h_{\Lambda}^{1/2}) h_{\Lambda}^{-1/2} (h_{\Lambda}^{(p)})^{1/2} & 0 \\ 0 & (h_{\Lambda}^{(p)})^{-1/2}  \coth( \beta h_{\Lambda}^{1/2}) h_{\Lambda}^{1/2} (h_{\Lambda}^{(p)})^{-1/2} \end{array} \right) \, 
\end{equation}
and it satisfies $2C_{\beta} = M_{\beta} - i J$ for the corresponding thermal state covariance matrix.
It is clear that both $\omega$ and $\omega_{\beta}$ define quasi-free functionals on $\mathcal{W}_{\Lambda}$
in the sense of (\ref{quasi-free}).

Using (\ref{positivity}), one also readily checks that the quasi-free functionals introduced above, i.e. both $\omega$ and $\omega_{\beta}$,
are states. In fact, let
\begin{equation}
R = \left( \begin{array}{cc} M_1^{1/2} & 0 \\ 0 & M_2^{1/2} \end{array} \right)
\end{equation}
where $M_1$ and $M_2$ are, respectively, the upper left and lower right entries in the matrix $M$ from (\ref{Mmat}).
It is clear that $R$ is symplectic, i.e.,
\begin{equation}
R^T J R = J \, 
\end{equation}
and therefore, $R$ leaves the symplectic form invariant, i.e. 
\begin{equation}
\sigma(f,g) = (Jf,g) = (JRf, Rg) = \sigma(Rf, Rg)
\end{equation}
In this case,
\begin{equation}
\sigma(f,g)^2 = |(JRf, Rg)|^2 \leq \| Rf \|^2 \|Rg \|^2 = (f, Mf) (g, Mg)
\end{equation}
and (\ref{positivity}) holds with $s(f,f) = (f, Mf)$. With the relation
\begin{equation}
\coth(x) = 1 + \frac{2}{e^{2x} -1} \, ,
\end{equation}
it is clear that $M \leq M_{\beta}$, and thus the above argument proves that all of these functionals are states.

In fact, the quasi-free states for finite oscillator systems are always given by a density matrix on the Hilbert space as in (\ref{densitymatrix}).

As indicated previously, we are mainly interested in the logarithmic negativity
associated to the above states. Motivated by calculations
in Appendix~\ref{sec:partial_transpose}, see e.g. (\ref{eq:weyltranspose}), we define a partially transposed ground state functional
by setting 
\begin{equation} 
\omega^{T_1}(W(f))= e^{-\frac{1}{4} (f,\tilde M f)} \quad \mbox{ with } \quad
\tilde M = \mat{\idty}{0}{0}{\mathbb{P}} M \mat{\idty}{0}{0}{\mathbb{P}}\, ,
\label{transposed_functional}\end{equation}
$M$ is as in (\ref{Mmat}), $\Lambda_0 \subset \Lambda$ fixed, and $\mathbb{P}$ is the diagonal matrix with
\begin{equation} \label{eq:diagP}
\mathbb{P}_{xx}= \left\{ \begin{array}{rl} -1 & \mbox{if } x \in \Lambda_0 \\ 1
& \mbox{if } x \not\in \Lambda_0\, . \end{array} \right.
\end{equation}
This functional is quasi-free, self-adjoint ($\omega^{T_1}(W(f)^*) = \omega^{T_1}(W(f))$ due to (\ref{eq:invo})), and normalized such that
$\omega^{T_1}(\idty)=1$, but in general, it is not positive and therefore not
a state. A partially transposed thermal state functional, $\omega_{\beta}^{T_1}$, is 
analogously defined by replacing $M$ above with $M_{\beta}$ as in (\ref{eq:mbeta}).  
 
%%%%%%%%%%%%%%%%%%%%%%
%
%
%          On diagonalizing quasi-free functionals
%
%
%%%%%%%%%%%%%%%%%%

\subsection{Diagonalizing quasi-free functionals} \label{sec:will}
In this section, we return to the general setting of Section ~\ref{sec:set-up} to discuss 
the diagonalization of quasi-free functionals. Recall that for any real-linear space $\mathcal{D}$
equipped with a non-degenerate, symplectic form $\sigma$, a quasi-free functional
$\omega$ on the Weyl algebra $\mathcal{W} = \mathcal{W}(\mathcal{D})$ has the form
\begin{equation}
\omega(W(f)) = e^{ir(f) - \frac{1}{4} s(f,f)} \quad \mbox{for all } f \in \mathcal{D} \, ,
\end{equation} 
where $r$ is a real-linear functional and $s$ is a symmetric, real bilinear form.
As is the case in our examples, we will assume $\mathcal{D}$ is finite dimensional.
In fact, without loss of generality, we will assume $\mathcal{D} = \ell^2( \Lambda)$
for some finite set $\Lambda$. By diagonalizing $\omega$, we mean finding an automorphism $\alpha$ of the 
Weyl algebra (or, equivalently, a unitary or anti-unitary  transformation on Fock space) 
for which 
\begin{equation} \label{eq:diagfun}
\omega(\alpha(W(f)))=\prod_{k=1}^n e^{-\frac{1}{2}\lambda_k |f(k)|^2}
\end{equation}
for some $\lambda_k\in \mathbb{R}$ and $n = | \Lambda|$. This will enable us to calculate explicitly unitarily
invariant quantities such as the entropy or the $p$-norm of the trace class operator
associated with $\omega$.

We start by noting that, without loss of generality, we may assume that $r=0$. This is because
for any such real-linear functional $r$, there exists $g \in \mathcal{D}$ such that $r(f)=\sigma(f,g)$.
In this case, the automorphism $\tilde{\alpha}$ defined by
\begin{equation}
\tilde{\alpha}(W(f)) = W(g)^* W(f) W(g) = W(-g)W(f) W(g)=e^{-i\sigma(f,g)}W(f),
\end{equation}
where we have used the Weyl relations (\ref{eq:invo}) and (\ref{eq:weylrel}), satisfies
\begin{equation}
\omega(\tilde{\alpha}(W(f)))=e^{-\frac{1}{4}s(f,f)} \, .
\end{equation}

For the next step, we diagonalize the form $s$. 
This is done by choosing a basis in $\mathcal{D}$ and representing $s$
in terms of a $2n \times 2n$ real, symmetric matrix $M$. A well-known version of
Williamson's Theorem (see, e.g., \cite{SCS99}) then provides the existence of a symplectic matrix $S$
which diagonalizes $M$; we state this as Proposition~\ref{prop:williamson} below. 
As is proven e.g. in Theorem 5.2.8 of \cite{BratRob}, such an $S$ induces an automorphism
on $\mathcal{W}$, in terms of which (\ref{eq:diagfun}) is then clear.

\begin{prop}\label{prop:williamson}
Let $n \geq 1$ and suppose that $M$ is a real symmetric positive definite $2n\times 2n$ matrix.
Then, there exists a symplectic $S$ such that
\begin{equation}
S^T M S = \mat{\mathcal{L}}{0}{0}{\mathcal{L}}
\end{equation}
where $\mathcal{L}$ is a diagonal matrix with entries $\lambda_k > 0$
for $k=1,\ldots, n$. The numbers $\lambda_k$ are the positive eigenvalues of
$iM^{1/2}JM^{1/2}$; they are also the positive imaginary parts of the eigenvalues of $MJ$.
If, and only if, in addition,
\begin{equation}
M+ iJ\geq 0
\label{add-cond}\end{equation}
we have $\lambda_k \geq 1$ for all $k=1,\ldots, n$. 
\end{prop}

\begin{proof}
The existence of a symplectic matrix $S$ that diagonalizes $M$ in the prescribed
fashion is a special case of Wiliamson's Theorem \cite{williamson1936}, a particularly simply
proof of which is given in \cite{SCS99}. This includes the statement that the $\lambda_k$ are
positive. It follows from the proof of Theorem 3 in \cite{SCS99}
that the $\lambda_k$ are the positive eigenvalues of $iM^{1/2}JM^{1/2}$.

If the additional condition (\ref{add-cond}) is satisfied, we have
\begin{equation}
0\leq S^T(M+ iJ)S = \mat{\mathcal{L}}{0}{0}{\mathcal{L}} +  iJ
=\mat{\mathcal{L}}{-i\idty}{i\idty}{\mathcal{L}}
\end{equation}
After reordering, the RHS is block diagonal with two-dimensional blocks of the
form
\begin{equation}
\mat{\lambda_k}{- i}{ i}{\lambda_k}
\end{equation}
which is non-negative definite if and only if $\lambda_k\geq 1$.
\end{proof}

\begin{rem} \label{rem:sympev}
In the special case where $M$ is block diagonal with two $n\times n$ blocks, i.e., of the form
\begin{equation}
M=\mat{M_1}{0}{0}{M_2} \, ,
\label{blockform}\end{equation}
then the hermitian matrix $iM^{1/2}JM^{1/2}$ takes the form
\begin{equation}
iM^{1/2}JM^{1/2}= i\mat{0}{-M_1^{1/2}M_2^{1/2}}{M_2^{1/2}M_1^{1/2}}{0}
\end{equation}
the square of which is
\begin{equation}
(iM^{1/2}JM^{1/2})^2=\mat{M_1^{1/2}M_2M_1^{1/2}}{0}{0}{M_2^{1/2}M_1M_2^{1/2}}\,.
\end{equation}
Thus, in this case, the symplectic eigenvalues of $M$ can also be found e.g. as the positive square roots
of the eigenvalues of $M_1^{1/2}M_2M_1^{1/2}$. Since $M_2^{1/2}M_1M_2^{1/2}$
and $M_1M_2$, are both similar to $M_1^{1/2}M_2M_1^{1/2}$ they have the same eigenvalues, and
therefore, they could, as well, be used in determining these symplectic eigenvalues.
\end{rem}

%%%%%%%%%%%%%%%%%%%%%%
%
%
%          1 norm calc.
%
%
%%%%%%%%%%%%%%%%%%

\subsection{A formula for the logarithmic negativity} \label{sec:1norm}

Now, we can combine the information of the previous subsections to prove
the expression for the logarithmic negativity of ground and thermal states in 
Theorem~\ref{thm:logneg}. The essential observation is that $\tilde M$, 
see e.g (\ref{transposed_functional}), is real symmetric and positive whenever 
$M$ is, and therefore, the first part of Proposition \ref{prop:williamson} applies. 
This symplectically diagonalizes $\tilde M$. The second part of the above-mentioned 
proposition does {\em not} apply, however, because 
\begin{equation}
\mat{\idty}{0}{0}{\mathbb{P}}\mat{0}{-\idty}{\idty}{0} \neq
\mat{0}{-\idty}{\idty}{0} \mat{\idty}{0}{0}{\mathbb{P}},
\end{equation}
i.e., the partial transpose does not preserve the symplectic form $J$. 
We begin with a statement of the main result.

\begin{thm} \label{thm:logneg} Fix $\Lambda_0 \subset \Gamma$ finite. 
For any finite $\Lambda \subset \Gamma$ with $\Lambda_0 \subset \Lambda$ and $H_{\Lambda}$
as in (\ref{eq:localham}) satisfying Assumption 2.1, we have that the logarithmic negativity associated
to the ground state, respectively thermal state, and the decomposition in (\ref{decomp}) is given by
\begin{equation} \label{lognegprop}
\mathcal{N}( \rho) =\frac{1}{2} {\rm Tr} \left[ P^+ \log(L^{-1}) \right] 
\end{equation}
where
\begin{equation} \label{defl}
L = M_1^{1/2} \mathbb{P} M_2 \mathbb{P} M_1^{1/2} \, ,
\end{equation}
$P^+$ is the orthogonal projection onto the subspace where $L \leq \idty$, $\mathbb{P}$ is the
diagonal matrix with 
 \begin{equation} \label{eq:diagP}
\mathbb{P}_{xx}= \left\{ \begin{array}{rl} -1 & \mbox{if } x \in \Lambda_0 \\ 1
& \mbox{if } x \in \Lambda \setminus \Lambda_0 \, , \end{array} \right.
\end{equation}
and for the ground state we have
\begin{equation}
M_1 =  (h_{\Lambda}^{(p)})^{1/2} (h_{\Lambda})^{-1/2} (h_{\Lambda}^{(p)})^{1/2} \quad 
\mbox{and}
 \quad M_2 = (h_{\Lambda}^{(p)})^{-1/2} (h_{\Lambda})^{1/2} (h_{\Lambda}^{(p)})^{-1/2} \, ,
\end{equation}
while for the thermal state at inverse temperature $\beta$ we have
\begin{align}
M_1& =  (h_{\Lambda}^{(p)})^{1/2} (h_{\Lambda})^{-1/2} 
\coth(\beta  (h_{\Lambda})^{1/2})(h_{\Lambda}^{(p)})^{1/2}
\label{thermalM1}\\
M_2& = (h_{\Lambda}^{(p)})^{-1/2} (h_{\Lambda})^{1/2} \coth(\beta  (h_{\Lambda})^{1/2}) (h_{\Lambda}^{(p)})^{-1/2} \, 
\label{thermalM2}
\end{align}
with $h_{\Lambda} = (h_{\Lambda}^{(p)})^{1/2} h_{\Lambda}^{(q)} (h_{\Lambda}^{(p)})^{1/2}$.
\end{thm}

We prepare the proof by a lemma in which we consider the Weyl algebra over the one-dimensional vector space ${\mathcal D} = \C$, denoted by ${\mathcal W}(\C)$, and generated by 
\begin{equation}
W(z) = \exp{\left(i [({\rm Re}\,z) q + ({\rm Im}\,z)p]\right)} = \exp{\left(\frac{i}{\sqrt{2}}(\overline{z} a + z a^*)\right)}, \quad z\in \C,
\end{equation}
where $a=(q+ip)/\sqrt{2}$ and $a^* = (q-ip)/{\sqrt 2}$.

\begin{lemma} \label{lem:gau} Fix $\lambda >0$. There exists a unique, self-adjoint, trace-class operator $\rho_{\lambda}$ on $L^2( \mathbb{R})$ for which
\begin{equation}
{\rm Tr} \left[ \rho_{\lambda} W(z) \right] = e^{-\frac{1}{4}\lambda|z|^2} \quad \mbox{for all $z\in \C$} \, .
\end{equation}
Moreover,
\begin{equation} \label{gau1norm}
\Vert \rho_\lambda\Vert_1 = \left\{ \begin{array}{cl} 1 & \mbox{if } \lambda \geq 1 \\
\frac{1}{\lambda} & \mbox{if } \lambda < 1 \, .\end{array} \right.
\end{equation}
\end{lemma} 

\begin{proof}
Denote by $\left\{| n \rangle\mid n=0,1,\ldots\right\}$ the orthonormal basis of $L^2(\mathbb{R})$
given by the eigenvectors of the standard harmonic oscillator Hamiltonian $a^*a +\frac{1}{2}$.
In this basis $a^*$ and $a$ correspond to creation and annihilation operators,
\begin{equation}
a^*| n \rangle = \sqrt{n+1}|n+1 \rangle, \mbox{for\ } n \geq 0, 
\mbox{ and } a| n\rangle = \sqrt{n}|n-1 \rangle, \mbox{for\ } n \geq 1.
\end{equation}
Using the commutation relation $[a^*,a]=\idty$, one readily verifies
\begin{equation} \label{weylexp}
\langle n | W(z) | n \rangle = e^{\frac{|z|^2}{4}} \sum_{m \geq 0} \frac{(- |z|^2)^m}{2^m (m!)^2} \frac{(n+m)!}{n!} \, ,
\end{equation}
see e.g. (XII.55) of \cite{Messiah1999} for more details. Note that the formal Taylor expansion used here can be justified by an analytic vector argument, e.g.\ Theorem~8.30 in \cite{Weidmann}. Since $\lambda>0$, the number
\begin{equation}
\alpha = \alpha(\lambda) =  \frac{\lambda -1}{\lambda +1},
\end{equation}
satisfies $-1 < \alpha < 1$ and therefore, the operator $\rho_\lambda$ defined by
\begin{equation}
\rho_\lambda |n \rangle = (1-\alpha)\alpha ^n |n\rangle \quad \mbox{for } n \geq 0,
\end{equation} 
is clearly self-adjoint and trace class. A well-known calculation (e.g., see again \cite{Messiah1999}) using the identity
\begin{equation}
\frac{1}{(1-x)^{m+1}} = \sum_{n \geq 0} \frac{(n+m)!}{m! n!} x^n
\end{equation}
gives 
\begin{eqnarray}
\Tr [ \rho_\lambda W(z) ]  & = &(1-\alpha) \sum_{n \geq 0} \alpha^n
 \langle n|W(z) | n \rangle \\
& = & (1- \alpha) e^{\frac{|z|^2}{4}} \sum_{n \geq 0} \sum_{m \geq 0} \alpha^n \frac{(-|z|^2)^m}{2^m (m!)^2} \frac{(n+m)!}{n!} \nonumber \\
& = & (1- \alpha) e^{\frac{|z|^2}{4}} \sum_{m \geq 0}  \frac{(-|z|^2)^m}{2^m (m!)}  \sum_{n \geq 0} \frac{(n+m)!}{m! n!} \alpha^n \nonumber \\
& = & e^{\frac{|z|^2}{4}} \sum_{m \geq 0}  \frac{(-|z|^2)^m}{2^m (m!)} \frac{1}{(1- \alpha)^m} \nonumber \\
& = & e^{\frac{|z|^2}{4}} e^{\frac{-|z|^2}{2(1- \alpha)}} \nonumber\\
&=&e^{- \frac{\lambda}{4}|z|^2} \, . \nonumber
\end{eqnarray}
{F}rom the explicit form of $\rho_{\lambda}$, it is also clear that
\begin{equation}
\Vert \rho_\lambda\Vert_1  = \frac{1-\alpha}{1-\vert \alpha\vert} \, ,
\end{equation}
and therefore, we obtain (\ref{gau1norm}).

Finally, the uniqueness of $\rho_{\lambda}$ follows from Lemma~\ref{lem:weakdense}.
\end{proof}

{\it Proof of Theorem~\ref{thm:logneg}:} 

We verify (\ref{lognegprop}) for the ground state $\rho$. The argument for
thermal states $\rho_{\beta}$ follows similarly.

According to $L^2(\mathbb{R}^{\Lambda}) = L^2(\mathbb{R}^{\Lambda_0}) \otimes L^2(\mathbb{R}^{\Lambda \setminus \Lambda_0})$, Weyl operators $W(f)$, $f\in \ell^2(\Lambda)$, can be decomposed as
\begin{equation}
W(f) = W(f^{(1)}) \otimes W(f^{(2)}),
\end{equation}
where $f^{(1)} = f|_{\Lambda_0}$, $f^{(2)} = f|_{\Lambda \setminus \Lambda_0}$. Thus, using (\ref{eq:weyltranspose}), the partial transpose with respect to this decomposition becomes
\begin{equation}
W(f)^{T_1} = W(f^{(1)})^T \otimes W(f^{(2)}) = W\left(\overline{f^{(1)}}\right) \otimes W(f^{(2)}) = W \left( \mat{\idty}{0}{0}{\mathbb{P}} f \right),
\end{equation}
where the last step uses the identification (\ref{real_basis}) and the definition of $\mathbb{P}$. It now follows from (\ref{eq:Weylexp}) that
\begin{equation} \label{eq:transtrace}
{\rm Tr} \,\rho\, W(f)^{T_1} = \exp{\left( - \frac{1}{4} \left(  \mat{\idty}{0}{0}{\mathbb{P}} f, M  \mat{\idty}{0}{0}{\mathbb{P}} f\right) \right)} = e^{-\frac{1}{4}(f, \tilde{M}f)},
\end{equation}
where $\tilde{M}$ is the real, symmetric, positive matrix defined in (\ref{transposed_functional}). By Proposition~\ref{prop:williamson} there exists a symplectic matrix $S$ such that
\begin{equation}
S^T \tilde{M} S = \mat{\mathcal{L}}{0}{0}{\mathcal{L}},
\end{equation}
where $\mathcal{L}$ is diagonal with entries $\lambda_j>0$, $j=1,\ldots,|\Lambda|$, the symplectic eigenvalues of $\tilde{M}$. With $f=Sg$ in (\ref{eq:transtrace}) we get
\begin{eqnarray}
{\rm Tr}\,\rho\,W(Sg)^{T_1} & = & e^{-\frac{1}{4}(g, S^T \tilde{M} Sg)} = \prod_j e^{-\frac{1}{4}\lambda_j |g_j|^2} \\
& = & \prod_j {\rm Tr}\,\rho_{\lambda_j} W(g_j) = {\rm Tr} \left( \otimes_j \rho_{\lambda_j} \right) W(g), \nonumber
\end{eqnarray}
where $\rho_{\lambda_j}$ are the trace class operators on $L^2(\mathbb{R})$ introduced in Lemma~\ref{lem:gau}.

Letting $U$ be the unitary that implements $S$ in the representation of the Weyl algebra
we are using (see \cite{Shale1962} or \cite{Bruneau2007}), i.e.,
\begin{equation}
W(Sf) = UW(f) U^*
\end{equation}
we observe that, for all $g\in \ell^2(\Lambda)$,
\begin{equation}
{\rm Tr}\,\rho\, W(Sg)^{T_1} = {\rm Tr}\, \left( \otimes_j \rho_{\lambda_j} \right) U^* W(Sg) U = {\rm Tr}\, U\left( \otimes_j \rho_{\lambda_j} \right) U^* W(Sg),
\end{equation}
or, as $S$ is invertible,
\begin{equation}
{\rm Tr}\, \rho\,W(f)^{T_1} = {\rm Tr} \,U \left( \otimes_j \rho_{\lambda_j} \right) U^*\,W(f) \quad \mbox{for all $f\in \ell^2(\Lambda)$}.
\end{equation}

Now Corollary~\ref{cor:ptlemma}  proves that $\rho^{T_1} = U (\otimes_j \rho_{\lambda_j}) U^*$. In particular, $\rho^{T_1}$ is trace class and
\begin{equation}
\|\rho^{T_1}\|_1 = \prod_j \|\rho_{\lambda_j}\|_1.
\end{equation}
Using (\ref{gau1norm}) to calculate the trace norm of $\rho_{\lambda_j}$ we get
\begin{equation} \label{eq:finalcalc}
\mathcal{N}(\rho) = \log \,\|\rho^{T_1}\|_1 = \sum_j \log \|\rho_{\lambda_j}\|_1 = \sum_{j:\lambda_j<1} \log \frac{1}{\lambda_j}.
\end{equation}
By Remark~\ref{rem:sympev} (applied to $\tilde{M}$) the $\lambda_j$ are the positive square roots of the eigenvalues of $L = M_1^{1/2} \mathbb{P} M_2 \mathbb{P} M_1^{1/2}$. This shows that the right hand side of (\ref{eq:finalcalc}) is equal to $\frac{1}{2} {\rm Tr} [P^+ \log(L^{-1})]$. \hfill \qed

%%%%%%%%%%%%%%%%%%%%%%%%%%
%
%
%	Area Law for the Log. Neg. + Proof of the Main Result
%
%
%
%%%%%%%%%%%%%%%%%%%%%%%%%%%

\section{An Area Law for the Logarthimic Negativity} \label{sec:proof}

The goal of this section is to prove Theorems~\ref{thm:gsal} and \ref{thm:tsal}.
This is done after establishing two deterministic facts. First, in Lemma~\ref{lem:logneg} we prove an
upper bound on the logarithmic negativity associated to ground and thermal states
of oscillator systems. The arguments here follow
\cite{Cramer2006} rather closely, however, we avoid making assumptions
on the spectral gap of the one-particle operators involved. Next, we prove a
simple geometric fact about the graphs we are considering in Lemma~\ref{lem:boundary}. The proofs of our main
results follow.

\begin{lemma} \label{lem:logneg} Fix $\Lambda_0 \subset \Gamma$ finite. 
For any finite $\Lambda \subset \Gamma$ with $\Lambda_0 \subset \Lambda$ and $H_{\Lambda}$
as in (\ref{eq:localham}) satisfying Assumption~\ref{ass:bobs}, we have the following bound on
the logarithmic negativity associated to the ground state, respectively thermal state, and the decomposition in (\ref{decomp}):
\begin{equation} \label{eq:lognegbd}
\mathcal{N}( \rho) \leq 2 \| M_1^{-1} \| \sum_{x \in \Lambda_0} \sum_{y \in \Lambda \setminus \Lambda_0} \left| \langle \delta_x,
M_2^{-1} \delta_y \rangle \right|,
\end{equation}
where for the ground state we have set
\begin{equation}
 M_1^{-1} = (h_{\Lambda}^{(p)})^{-1/2} (h_{\Lambda})^{1/2} (h_{\Lambda}^{(p)})^{-1/2} \quad \mbox{and} \quad M_2^{-1} =  (h_{\Lambda}^{(p)})^{1/2} (h_{\Lambda})^{-1/2} (h_{\Lambda}^{(p)})^{1/2} 
\end{equation}
while for the thermal state at inverse temperature $\beta$,
\begin{align}
M_1^{-1}& = (h_{\Lambda}^{(p)})^{-1/2} (h_{\Lambda})^{1/2} \tanh(\beta  (h_{\Lambda})^{1/2}) (h_{\Lambda}^{(p)})^{-1/2} \\
M_2^{-1}& =  (h_{\Lambda}^{(p)})^{1/2} (h_{\Lambda})^{-1/2} \tanh(\beta  (h_{\Lambda})^{1/2})(h_{\Lambda}^{(p)})^{1/2}\,
\end{align}
with $h_{\Lambda} = (h_{\Lambda}^{(p)})^{1/2} h_{\Lambda}^{(q)} (h_{\Lambda}^{(p)})^{1/2}$ in both cases.
\end{lemma}

\begin{proof}
Recall the results of Theorem~\ref{thm:logneg}; in particular, (\ref{lognegprop}) and (\ref{defl}).
Using the concavity of the logarithm, we immediately obtain the following bound:
\begin{equation}
2 \mathcal{N}(\rho)= \Tr P^+ \log L^{-1}\leq \Tr P^+ (L^{-1}-\idty)\, .
\end{equation}
Next, we note that $L^{-1}=  M_1^{-1/2}\mathbb{P}M_2^{-1}\mathbb{P}M_1^{-1/2}$ can be
written as $L^{-1}=A+B$ with
\begin{equation}
A= M_1^{-1/2}M_2^{-1}M_1^{-1/2} \quad \mbox{and} \quad B=M_1^{-1/2}\mathbb{P}[M_2^{-1},\mathbb{P}]M_1^{-1/2}.
\end{equation}
We claim $A\leq \idty$. In the ground state case, we have $M_2=M_1^{-1}$. This directly
implies $A=\idty$, and there is nothing to prove. In the thermal case, recall that $M_1$ and $M_2$
are given by (\ref{thermalM1}) and (\ref{thermalM2}), respectively. As a result, they are both positive matrices; hence, so too are
$L^{-1}$ and $A$. It is clear that
\begin{equation}
M_1^{-1}M_2^{-1} = (h_{\Lambda}^{(p)})^{-1/2} \tanh^2(\beta  (h_{\Lambda})^{1/2})  (h_{\Lambda}^{(p)})^{1/2} 
\end{equation}
and thus $A$ is similar to $\tanh^2(\beta h_\Lambda^{1/2})$, i.e.,

\begin{equation}
A= K \tanh^2(\beta h_\Lambda^{1/2}) K^{-1} \quad \mbox{with} \quad K= M_1^{1/2} (h_{\Lambda}^{(p)})^{-1/2}\,.
\end{equation}
We conclude that the spectrum of the positive operator $A$ is contained in $[0,1]$, and hence we have
$A\leq \idty$ as claimed. Using this fact we obtain
\begin{equation}
2 \mathcal{N}(\rho)\leq \Tr P^+ (A-\idty+B) \leq \Tr  P^+B\, .
\end{equation}
 To proceed, note that
\begin{align}
 \Tr P^+B &= \Tr P^+ M_1^{-1/2} \mathbb{P} [M_2^{-1},\mathbb{P}]  M_1^{-1/2} P^+  \\
 &\leq \Vert P^+ M_1^{-1/2} \Vert^2\, \Vert \mathbb{P} [M_2^{-1},\mathbb{P}] \Vert_1\nonumber\\
&\leq \Vert P^+ M_1^{-1/2} \Vert^2\, \sum_{x,y\in\Lambda}\vert \langle\delta_x,
\mathbb{P} [M_2^{-1},\mathbb{P}] \delta_y\rangle\vert. \nonumber
\end{align}
The first factor is of course bounded by $\Vert M_1^{-1}\Vert$. For the second factor we
used the fact that the 1-norm can be bounded by the sum of the absolute values of
the matrix elements in any basis. Observing that
\begin{equation}
\langle \delta_x, \mathbb{P} \left[ M_2^{-1}, \mathbb{P} \right] \delta_y \rangle = \left\{ \begin{array}{cc} 0 & \mbox{if } x, y \in \Lambda_0 \mbox{ or } x, y \in \Lambda \setminus \Lambda_0 , \\
-2 \langle \delta_x, M_2^{-1}\delta_y \rangle & \mbox{otherwise} . \end{array} \right.
\end{equation}
leads to the claimed result.
\end{proof}

As observed in \cite{Cramer2006}, if $h_\Lambda \geq c >0$, $P^+=0$ for
sufficiently small $\beta$, meaning that all entanglement vanishes at high temperatures. The previous result holds, however, at all positive temperatures.

For the sake of completeness, we prove a general fact about the graphs
we are considering.

\begin{lemma}\label{lem:boundary}
Let $\Lambda_0 \subset \Gamma$ be finite and suppose $\mu >0$ is such that
\begin{equation}
C_\mu:=\sup_{x\in\Gamma}\sum_{y\in\Gamma} e^{-\mu d(x,y)}<\infty\,.
\end{equation}
Then for any $\Lambda_0 \subset \Lambda \subset \Gamma$,
\begin{equation}
\sum_{x\in\Lambda_0}\sum_{y\in \Lambda\setminus\Lambda_0} e^{-\mu d(x,y)}
\leq C_\mu^2 \vert \partial \Lambda_0\vert
\end{equation}
where $\vert \partial \Lambda_0 \vert$ is as defined in (\ref{boundary}).
\end{lemma}

\begin{proof}
Note that for every $x\in \Lambda_0$ and $y\in \Lambda \setminus \Lambda_0$ there is at least one $u\in \partial \Lambda_0$ such that $d(x,u)+d(u,y)=d(x,y)$. Therefore
\begin{align}
\sum_{x\in\Lambda_0}\sum_{y\in \Lambda\setminus\Lambda_0} e^{-\mu d(x,y)}
& \le \sum_{u\in \partial \Lambda_0}\sum_{x\in\Lambda_0, y\in \Lambda\setminus\Lambda_0
\atop d(x,u)+d(u,y)=d(x,y)} e^{-\mu d(x,u)}e^{-\mu d(u,y)}\\
&\leq  \sum_{u\in \partial \Lambda_0} \sum_{x\in \Gamma, y\in \Gamma} e^{-\mu d(x,u)} e^{-\mu d(u,y)} \nonumber \\
&\leq  C_\mu^2 \vert \partial \Lambda_0\vert\,. \nonumber
\end{align} 
\end{proof}

We can now complete the proofs of Theorems~\ref{thm:gsal} and \ref{thm:tsal}.

\begin{proof}[Proof of Theorem~\ref{thm:tsal}]
Using Lemma~\ref{lem:logneg} and Assumption~\ref{ass:bobs}, it is clear that
\begin{equation}
\mathbb{E} \left( \mathcal{N}( \rho) \right) \leq C \sum_{x \in \Lambda_0} \sum_{y \in \Lambda \setminus \Lambda_0} \mathbb{E} \left( \left| \langle \delta_x,
M_2^{-1} \delta_y \rangle \right| \right)
\end{equation}
The quantity on the right-hand-side above is precisely what appears in (\ref{eq:tsloccond}). 
(\ref{tsal}) now follows from Lemma~\ref{lem:boundary}.
\end{proof}

\begin{proof}[Proof of Theorem~\ref{thm:gsal}]
The proof of Theorem~\ref{thm:gsal} follows as above,
after an application of Lemma~\ref{lem:purebd}.
\end{proof}

\appendix

%%%%%%%%%%%%%%%%%%%%%%%%%
%
%       Appendix:
%	
%
%%%%%%%%%%%%%%%%%%%%%%%%%%%%%

\section{On Partial Transposes and an Entanglement Bound}
\label{sec:partial_transpose}

The goal of this section is to collect some basic facts about partial transposes
valid in the context of infinite dimensional complex Hilbert spaces. We begin with
a discussion of general conjugations and their corresponding transposes in Section~\ref{subsec:conj}.
An example of the transpose of a Weyl operator with respect to the natural conjugation ends this
section. Partial transposes are defined in Section~\ref{subsec:part}, and a number of important properties,
each used in the main text, are proven in detail. Finally, Section~\ref{subsec:ebps} contains a proof, originally
given by Vidal and Werner in \cite{Vidal2002}, that the logarithmic negativity provides an upper bound on the von-Neuman entropy of the restriction of 
a pure state.

\subsection{Conjugations and Transposes} \label{subsec:conj}

First, recall the definition of the operator transpose with respect to a given conjugation in a Hilbert space.

Let $({\mathcal H}, \langle \cdot, \cdot \rangle)$ be a separable complex Hilbert space and $C:{\mathcal H} \to {\mathcal H}$ be a {\it conjugation}, i.e.\ $\langle Cf, Cg \rangle = \langle g, f \rangle$ for all $f, g \in {\mathcal H}$ and $C^2 = I$.

This implies that $C$ is anti-linear, i.e.\ $C(\alpha f + \beta g) = \bar{\alpha} Cf + \bar{\beta} Cg$ for all $\alpha, \beta \in \C$, $f,g \in {\mathcal H}$.

An important fact is that $C$ is a conjugation in ${\mathcal H}$ if and only if there exists an orthonormal basis (ONB) $\{e_k\}$ of ${\mathcal H}$ such that
\begin{equation} \label{eq:specialONB}
C \left( \sum_k x_k e_k \right) = \sum_k \bar{x}_k e_k \quad \mbox{for all $(x_k) \in \ell_2$.}
\end{equation}

To see this, first note that for an anti-linear mapping $C$ the property (\ref{eq:specialONB}) is equivalent to the existence of an ONB such that
\begin{equation} \label{eq:specialONB2}
Ce_k = e_k \quad \mbox{for all $k$}.
\end{equation}

It's easy to see that (\ref{eq:specialONB2}) is sufficient for $C$ to be a conjugation. On the other hand, if $C$ is a conjugation, then by Zorn's lemma there is a maximal orthonormal system $\{e_k\}$ with the property (\ref{eq:specialONB2}). To show that this is an ONB consider the orthogonal complement $D$ of the subspace spanned by  $\{e_k\}$. If $D$ were non-trivial, then it would contain a normalized vector $u$. One can now construct a normalized vector $v \in D$ with $Cv=v$, contradicting the maximality of $\{e_k\}$. This is done separately for the two cases $Cu \in \,\mbox{span}\{u\}$ and $Cu \not\in \, \mbox{span}\{u\}$. In the first case $v$ can be chosen as a suitable scalar multiple of $u$, in the second case one may choose $v=(u+Cu)/\|u+Cu\|$.

While many different conjugations exist on any given Hilbert space, it follows that any two conjugations $C$ and $\tilde{C}$ on ${\mathcal H}$ are unitarily equivalent. In fact, for an ONB $\{\tilde{e}_k\}$ such that $\tilde{C} \tilde{e}_k = \tilde{e}_k$, let $U$ be the unique unitary such that
\begin{equation} \label{eq:defU}
U \tilde{e}_k = e_k
\end{equation}
for all $k$. Then it is easily seen that
\begin{equation} \label{eq:equivconj}
\tilde{C} = U^* C U.
\end{equation}

For any $A\in B({\mathcal H})$, the bounded linear operators on ${\mathcal H}$, its {\it transpose} $A^T$ with respect to the conjugation $C$ is defined by
\begin{equation}
A^T = C A^* C.
\end{equation}
We have $A^T \in B({\mathcal H})$ with $\|A^T\| = \|A\|$. Other basic properties of transposes which we will use below are $(A^*)^T = (A^T)^*$ and that $U$ is unitary if and only if $U^T$ is unitary.

With an ONB $\{e_k\}$ associated with $C$ via (\ref{eq:specialONB}) (and only with such ONBs), the transpose with respect to $C$ is characterized by
\begin{equation} \label{eq:chartranspose}
\langle e_j, A^T e_k \rangle = \langle e_k, A e_j \rangle \quad \mbox{for all $j, k$.}
\end{equation}

For two conjugations $C$ and $\tilde{C}$ which are unitarily equivalent via $U$ as in (\ref{eq:equivconj}), a calculation shows that the corresponding transposes are unitarily equivalent via $U^TU$,
\begin{equation} \label{eq:equivtrans}
A^{\tilde{T}} = (U^TU)^* A^T (U^TU).
\end{equation}

\begin{ex}
If ${\mathcal H} = L^2(X,\mu)$ for a measure space $(X,\mu)$, then the {\it natural conjugation} on ${\mathcal H}$ is given by $Cf = \bar{f}$ and the associated transpose by $A^T f = \overline{A^* \bar{f}}$. An ONB $\{e_k\}$ of $L^2(X,\mu)$ satisfies (\ref{eq:specialONB}) if and only if all $e_k$ are $\mu$-almost everywhere real-valued. In the special case ${\mathcal H} = \ell^2(\Gamma)$ for a countable set $\Gamma$, this is true for the canonical basis $e_j$, $j\in \Gamma$.

If $\Lambda$ is finite and ${\mathcal H} = {\mathcal H}_{\Lambda} = L^2(\R^{\Lambda}, dq)$ as in (\ref{eq:L2spacedef}), then for each $f\in \ell^2(\Lambda)$ the Weyl operators are defined by (\ref{eq:concreteWeyl}). The Weyl operators are unitary and their transposes with respect to natural conjugation are given by
\begin{equation} \label{eq:weyltranspose}
(W(f))^T = W(\bar{f}).
\end{equation} 
This is seen by combining the Weyl relations (\ref{eq:weylrel}) with the facts
\begin{equation} \label{eq:multoptranspose}
\phi(q)^T = \phi(q), \quad \phi(p)^T = \phi(-p),
\end{equation}
where, as usual, $\phi(q)$ and $\phi(p)$ denote multiplication operators in $q$ and $p$ space, respectively. In particular, $\phi(p) = F^{-1} \phi(q) F$ for the Fourier transform $(Ff)(y) = (2\pi)^{-1/2} \int e^{-ixy} f(x)\,dx$. Verifying the second one of the identities (\ref{eq:multoptranspose}) uses that $FC =RCF$, where $Rf(x)=f(-x)$.

\end{ex}

\subsection{Partial Transposes} \label{subsec:part}

We will also need the concept of a {\it partial transpose} of suitable classes of linear operators on tensor product spaces. For this, let ${\mathcal H}_1$ and ${\mathcal H}_2$ be separable complex Hilbert spaces and ${\mathcal H} = {\mathcal H}_1 \otimes {\mathcal H}_2$. Let $C$ be a conjugation on ${\mathcal H}_1$ and $A \mapsto A^T$ transposition with respect to $C$ on $B({\mathcal H}_1)$.

For tensor products $A\otimes B$ with $A\in B({\mathcal H}_1)$ and $B\in B({\mathcal H}_2)$, we define a partial transpose with respect to the first component of the tensor product as
\begin{equation} \label{eq:parttrans1}
(A \otimes B)^{T_1} := A^T \otimes B = CA^*C \otimes B.
\end{equation}

For finite dimensional Hilbert spaces ${\mathcal H}_1$ and ${\mathcal H}_2$ one can uniquely extend the definition of partial transpose to general $S \in B({\mathcal H})$ by imposing that the mapping $S\mapsto S^{T_1}$ is linear from $B({\mathcal H})$ to $B({\mathcal H})$. However, in the cases relevant here the Hilbert spaces have infinite dimension. In this case, an attempt to linearly extend (\ref{eq:parttrans1}) to general $S\in B({\mathcal H})$ will have to allow {\it unbounded} partial transposes $S^{T_1}$. This is due to the fact (e.g.\ \cite{AndoSano}) that in the finite-dimensional case the norm of the mapping $S\mapsto S^{T_1}$ as a linear operator on $B({\mathcal H})$ is given by $\min({\rm dim}\,{\mathcal H}_1,\, {\rm dim}\,{\mathcal H}_2)$ and thus increases with the dimensions. From this it is not hard to construct an example of a bounded operator on an infinite dimensional space with unbounded partial transpose.

We can avoid dealing with these issues here as in all our applications partial transposes will only have to be considered for operators which fall in one of two special classes of bounded operators, whose partial transposes are easily seen to be bounded. One of these classes are products $A\otimes B$, which are covered by (\ref{eq:parttrans1}), so that $\|(A\otimes B)^{T_1}\| = \|A^T \otimes B\| = \|A \otimes B\|$.

Another convenient class are the Hilbert-Schmidt operators on ${\mathcal H}$, which we will denote by $B_2({\mathcal H})$. Thus, let $S\in B_2({\mathcal H})$ and $\{e_j\}$ be the ONB of ${\mathcal H}_1$ associated with $C$ via (\ref{eq:specialONB}). Also, let $\{f_k\}$ be any ONB of ${\mathcal H}_2$. Define a linear operator $S^{T_1}$ which acts on the basis vectors $e_j \otimes f_k$ of ${\mathcal H}$ as
\begin{equation} \label{eq:parttrans2}
S^{T_1} e_j \otimes f_k = \sum_{m,\ell} \langle e_j \otimes f_{\ell}, S e_m \otimes f_k \rangle e_m \otimes f_{\ell}.
\end{equation}
Note that the meaning of partial transpose is reflected in the fact that its matrix elements are
\begin{equation} \label{eq:ptmatrix}
\langle e_m \otimes f_{\ell}, S^{T_1} e_j \otimes f_k \rangle = \langle e_j \otimes f_{\ell}, S e_m \otimes f_k \rangle,
\end{equation}
which should be compared with (\ref{eq:chartranspose}). Up to rearrangement, $S^{T_1}$ has the same matrix-elements as $S$, meaning that the Hilbert-Schmidt norm is preserved under partial transposition:
\begin{eqnarray}
\|S^{T_1}\|_2^2 & = & \sum_{j,k} \|S^{T_1} e_j \otimes f_k\|^2 = \sum_{j,k,m,\ell} |\langle e_j \otimes f_{\ell}, Se_m \times f_k \rangle|^2 \\
& = & \sum_{m,k} \|S e_m \otimes f_k \|^2 = \|S\|_2^2. \nonumber
\end{eqnarray}
In particular, the operator $S^{T_1}$ defined by (\ref{eq:parttrans2}) on a tensor product basis has a unique extension to a bounded operator on ${\mathcal H}$.

One can check that the definitions (\ref{eq:parttrans1}) and (\ref{eq:parttrans2}) are consistent, meaning that for $S=A\otimes B$, $A\in B_2({\mathcal H}_1)$, $B\in B_2({\mathcal H}_2)$, both definitions of $S^{T_1}$ coincide. Another property which holds for both classes is that $(S^{T_1})^{T_1} = S$.

For the next result we denote the trace class operators on ${\mathcal H}$ by $B_1({\mathcal H})$.

\begin{lemma} \label{lem:ptlemma}
(a) Let $S \in B_1({\mathcal H})$ such that $S^{T_1} \in B_1({\mathcal H})$. Also, let $A \in B({\mathcal H}_1)$ and $B\in B({\mathcal H}_2)$. Then
\begin{equation} \label{eq:ptshift}
{\rm Tr}\,S^{T_1}(A\otimes B) = {\rm Tr}\, S(A\otimes B)^{T_1}.
\end{equation}

(b) If $S, R \in B_1(\mathcal{H})$ are such that
\begin{equation} \label{eq:weakptcheck}
{\rm Tr}\,S\,(A\otimes B)^{T_1} = {\rm Tr}\,R\,(A\otimes B)
\end{equation}
for all $A\in B(\mathcal{H}_1)$ and $B\in B(\mathcal{H}_2)$, then $R=S^{T_1}$ and, in particular, $S^{T_1} \in B_1(\mathcal{H})$.
\end{lemma}

\begin{proof}
(a) Expanding $Ae_{\ell}$ and $Bf_m$ with respect to the bases $(e_j)$ and $(f_k)$, respectively, we get
\begin{equation} \label{eq:ptcalc}
\langle e_{\ell} \otimes f_m, S^{T_1} (A \otimes B) e_{\ell} \otimes f_m \rangle =  \sum_{j, k} \langle e_j, A e_{\ell} \rangle \langle f_k, B f_m \rangle \langle e_{\ell} \otimes f_m, S^{T_1} e_j \otimes f_k \rangle.
\end{equation}
Using (\ref{eq:ptmatrix}) as well as $\langle e_j, A e_{\ell} \rangle = \overline{ \langle e_j, (A^T)^* e_{\ell} \rangle}$, we see that the right hand side of (\ref{eq:ptcalc}) is equal to
\begin{equation}
\langle ((A^T)^* e_{\ell}) \otimes f_m, S e_{\ell} \otimes (Bf_m) \rangle = \langle e_{\ell} \otimes f_m, (A^T \otimes I)S(I\otimes B) e_{\ell} \otimes f_m \rangle.
\end{equation}
This allows to conclude by using cyclicity of the trace,
\begin{equation}
{\rm Tr}\,S^{T_1}(A \otimes B) = {\rm Tr}\,(A^T \otimes I)S(I\otimes B) = {\rm Tr}\, S(A^T \otimes B).
\end{equation}

(b) Note that it is not a priori clear that $S^{T_1} \in B_1(\mathcal{H})$. However, as $S\in B_1(\mathcal{H}) \subset B_2(\mathcal{H})$, we know from the above that $S^{T_1} \in B_2(\mathcal{H})$. Thus it is a bounded operator on $\mathcal{H}$ and therefore characterized by its matrix elements $\langle e_k \otimes f_m, S^{T_1} e_j \otimes f_{\ell} \rangle$.

Choosing $A=|e_j\rangle \langle e_k|$ and $B=|f_{\ell}\rangle \langle f_m|$ we find that ${\rm Tr}\,R(A\otimes B) = \langle e_k \otimes f_m, R e_j \otimes f_{\ell} \rangle$, while ${\rm Tr}\,S(A\otimes B)^{T_1} = \langle e_k \otimes f_m, S^{T_1} e_j \otimes f_{\ell} \rangle$. We conclude from (\ref{eq:weakptcheck}) that $R$ and $S^{T_1}$ have the same matrix elements and thus coincide.
\end{proof}

For the special case of $L^2$-spaces, it suffices to verify (\ref{eq:weakptcheck}) for Weyl operators:

\begin{cor} \label{cor:ptlemma}
Let $\mathcal{H}_1 = L^2(\mathbb{R}^{\Lambda_1})$, $\mathcal{H}_2 = L^2(\mathbb{R}^{\Lambda_2})$, $\mathcal{H} = L^2(\mathbb{R}^{\Lambda_1 \cup \Lambda_2})$, and let $C \mapsto C^{T_1}$ denote the partial transpose with respect to natural conjugation in the first component $\mathcal{H}_1$ of $\mathcal{H}$. If $S, R \in B_1(\mathcal{H})$ are such that
\begin{equation} \label{eq:ptcheck2}
{\rm Tr}\,S\, W(f)^{T_1} = {\rm Tr}\,R\,W(f)
\end{equation}
for all $f \in \ell^2(\Lambda_1 \cup \Lambda_2)$, then $R= S^{T_1}$.
\end{cor}

\begin{proof}
This can be reduced to Lemma~\ref{lem:ptlemma}(b) using the arguments in the proof of Lemma~\ref{lem:weakdense}. First, (\ref{eq:ptcheck2}) implies ${\rm Tr}\,S(W_1 \otimes W_2)^{T_1} = {\rm Tr}\,R(W_1\otimes W_2)$ for all $W_1 \in \mathcal{W}_{\Lambda_1}$, $W_2 \in \mathcal{W}_{\Lambda_2}$. 

As $\mathcal{W}_{\Lambda_j}$ is weakly dense in $B(\mathcal{H}_{\Lambda_j})$, $j=1,2$, for all $A\in B(\mathcal{H}_{\Lambda_1})$, $B\in B(\mathcal{H}_{\Lambda_1})$ there are $W_{j,n} \in \mathcal{W}_{\Lambda_j}$, $j=1,2$, $n=1,2,\ldots$, such that $W_{1,n} \stackrel{w}{\to} A$ and $W_{2,n} \stackrel{w}{\to} B$. Weak convergence is conserved under transposition and tensor products, i.e.\
\begin{equation}
W_{1,n} \otimes W_{2,n} \stackrel{w}{\to} A \otimes B, \quad W_{1,n}^T \otimes W_{2,n} \stackrel{w}{\to} A^T \otimes B.
\end{equation}
We conclude that (\ref{eq:weakptcheck}) holds. Thus Lemma~\ref{lem:ptlemma}(b) completes the proof.
\end{proof}

Unitary equivalence of transposes (\ref{eq:equivtrans}) with respect to different conjugations extends to partial transposes:

\begin{lemma} \label{lem:equivpt}
Let $C$ and $\tilde{C}$ be conjugations in ${\mathcal H}_1$, which are unitarily equivalent via $U$ as in (\ref{eq:equivconj}). Denote the corresponding transpose maps in $B({\mathcal H}_1)$ by $A\mapsto A^T$ and $A\mapsto A^{\tilde{T}}$. Let $S \in B_2({\mathcal H})$ or $S=A \otimes B$, $A\in B({\mathcal H}_1)$, $B\in B({\mathcal H}_2)$. Then $S^{T_1}$ and $S^{\tilde{T}_1}$ are unitarily equivalent via $V=(U^TU) \otimes I$,
\begin{equation} \label{eq:equivpt}
S^{\tilde{T}_1} = V^* S^{T_1} V.
\end{equation}
\end{lemma}

\begin{proof}
It is very likely that (\ref{eq:equivpt}) will hold for larger classes of operators $S$ if one goes through the effort of properly extending the definition of partial transpose, but it is most convenient here to merely check (\ref{eq:equivpt}) separately for the two cases covered by (\ref{eq:parttrans1}) and (\ref{eq:parttrans2}).

For tensor products $S = A \otimes B$, (\ref{eq:equivpt}) follows immediately from (\ref{eq:equivtrans}).

For $S\in B_2({\mathcal H})$ we need to verify that
\begin{equation} \label{eq:checkthis}
\langle \tilde{e}_m \otimes f_{\ell}, S^{\tilde{T}_1} \tilde{e}_j \otimes f_k \rangle = \langle \tilde{e}_m \otimes f_{\ell}, V^* S^{T_1} V \tilde{e}_j \otimes f_k \rangle
\end{equation}
for all $m$, $\ell$, $j$, $k$, using the ONBs introduced above. By (\ref{eq:defU}) and (\ref{eq:ptmatrix}) the left hand side of (\ref{eq:checkthis}) is equal to
\begin{equation} \label{eq:lhsofcheckthis}
\langle e_j \otimes f_{\ell}, (U \otimes I) S (U^* \otimes I) e_m \otimes f_k \rangle.
\end{equation}
For the right hand side of (\ref{eq:checkthis}) one gets
\begin{eqnarray}
\langle (U^T e_m) \otimes f_{\ell}, S^{T_1} (U^T e_j) \otimes f_k \rangle
& = & \sum_{r,s} \overline{\langle e_m, Ue_r \rangle} \langle e_j, Ue_s \rangle \langle e_r \otimes f_{\ell}, S^{T_1} e_s \otimes f_k \rangle \\
& = & \sum_{r,s} \langle e_r, U^* e_m \rangle \overline{\langle e_s, U^* e_j \rangle} \langle e_s \otimes f_{\ell}, S e_r \otimes f_k \rangle \nonumber \\
& = & \langle (U^* e_j) \otimes f_{\ell}, S(U^* e_m) \otimes f_k \rangle. \nonumber
\end{eqnarray}
The latter coincides with (\ref{eq:lhsofcheckthis}).

\end{proof}

A consequence of Lemma~\ref{lem:equivpt} which we will use later is that $S^{T_1} \in B_1({\mathcal H})$ if and only if $S^{\tilde{T}_1} \in B_1({\mathcal H})$ and that in this case
\begin{equation} \label{eq:tracestability}
\|S^{T_1}\|_1 = \| S^{\tilde{T}_1} \|_1.
\end{equation}

\subsection{Entanglement bounds for pure states} \label{subsec:ebps}
In this final section, we prove - mainly for the sake of completeness - a result
from \cite{Vidal2002}, see Lemma~\ref{lem:purebd} below. 

Let us briefly recall some basic definitions. Let $\mathcal{H}_1$ and $\mathcal{H}_2$ be 
separable complex Hilbert spaces and set $\mathcal{H} = \mathcal{H}_1 \otimes \mathcal{H}_2$. 
For any normalized vector $\Omega \in \mathcal{H}$, denote by $\rho$ the
orthogonal projection onto $\Omega$. By $\rho_1$, we will denote the restriction of
$\rho$ to $\mathcal{H}_1$, i.e.
\begin{equation}
\rho_1 = {\rm Tr}_{\mathcal{H}_2} \rho
\end{equation}
In the result we prove below, two quantities are of interest. The first, $S(\rho_1)$, is
the von Neumann entropy of this restriction:
\begin{equation}
S( \rho_1) = - {\rm Tr} \,\rho_1 \log( \rho_1) \, .
\end{equation}
The second is the logarithmic negativity of $\rho$ with respect to the
above decomposition. To define it, we chose a conjugation on $\mathcal{H}_1$ and set 
\begin{equation}
\mathcal{N}( \rho) = \log( \| \rho^{T_1} \|_1 ) \, 
\end{equation}
where $\rho^{T_1}$ is the partial transpose of $\rho$ with respect to this conjugation. By (\ref{eq:tracestability}) this definition is independent of the choice of conjugation. Note that
if $\rho^{T_1} \notin B_1( \mathcal{H})$, then the above is interpreted as infinite. The main result is
as follows:

\begin{lemma}[Vidal-Werner] \label{lem:purebd} Under the assumptions above,
\begin{equation} \label{eq:SNbd}
S(\rho_1) \leq \mathcal{N}(\rho) \, .
\end{equation}
\end{lemma}
\begin{proof}
We begin by writing $\Omega$ in terms of its Schmidt decomposition, i.e.,
\begin{equation}
\Omega = \sum_{\alpha} c_{\alpha} | e_{\alpha} \otimes f_{\alpha} \rangle
\end{equation}
where each of $\{ e_{\alpha} \} \subset \mathcal{H}_1$ and $\{ f_{\beta} \} \subset \mathcal{H}_2$ are orthonormal sets
indexed by a common, countable set. Without loss of generality, we will take $c_{\alpha}>0$. By normalization,
it is clear that  $\sum_{\alpha}c_{\alpha}^2 = 1$. Both quantities of interest in (\ref{eq:SNbd}) can be calculated
in terms of these coefficients $c_{\alpha}$. 

The restriction $\rho_1$ can be computed as
\begin{equation}
\rho_1 = \sum_{\alpha} c_{\alpha}^2 | e_{\alpha} \rangle \langle e_{\alpha} |
\end{equation}
from which it is clear that the $c_{\alpha}^2$ are the eigenvalues of $\rho_1$. In this case,
\begin{eqnarray} \label{eq:entropybd}
S(\rho_1)  = - \sum_{\alpha} c_{\alpha}^2 \ln( c_{\alpha}^2) & = & 2 \sum_{\alpha} \ln \left( \frac{1}{c_{\alpha}} \right) c_{\alpha}^2 \\
& \leq & 2 \ln \left( \sum_{\alpha} \frac{1}{c_{\alpha}} c_{\alpha}^2 \right) = 2 \ln \left( \sum_{\alpha} c_{\alpha} \right) \nonumber
\end{eqnarray}
the inequality following from Jensen's inequality. Without further assumptions, it is possible that the
right hand side is infinite.

For the logarithmic negativity, we only consider the case that $\| \rho^{T_1} \|_1 < \infty$; otherwise the
bound (\ref{eq:SNbd}) is trivial. 
Extend the orthonormal set $\{ e_{\alpha} \} \subset \mathcal{H}_1$ to an orthonormal basis, select
a conjugation $\tilde{C}$ leaving this orthonormal basis invariant, and denote by $\tilde{T}_1$ the
corresponding partial transpose. As remarked above, we may use any partial transpose to calculate the logarithmic negativity. We will use $\tilde{T}_1$.

As is clear from (\ref{eq:ptmatrix}), we have that
\begin{equation}
\rho^{\tilde{T}_1} = \sum_{\alpha, \beta} c_{\alpha} c_{\beta} | e_{\beta} \otimes f_{\alpha} \rangle \langle e_{\alpha} \otimes f_{\beta} | \, .
\end{equation}
Now, introduce a linear mapping $F: \mathcal{H} \to  \mathcal{H}$ by declaring that
\begin{equation}
F \left(  | e_{\alpha} \otimes f_{\beta} \rangle \right) =  | e_{\beta} \otimes f_{\alpha} \rangle
\end{equation}
for all $\alpha, \beta$ corresponding to vectors in the Schmidt decomposition and extend by the identity otherwise.
Then $F$ is a unitary and furthermore,
\begin{equation}
\rho^{\tilde{T}_1} = \sum_{\alpha, \beta} c_{\alpha} c_{\beta} | e_{\beta} \otimes f_{\alpha} \rangle \langle e_{\alpha} \otimes f_{\beta} |  =
\sum_{\alpha, \beta} c_{\alpha} c_{\beta} F \left( | e_{\alpha} \otimes f_{\beta} \rangle \right) \langle e_{\alpha} \otimes f_{\beta} | = F( C_1 \otimes C_2)\,,
\end{equation}
where we have set
\begin{equation}
C_1 = \sum_{\alpha} c_{\alpha} | e_{\alpha} \rangle \langle e_{\alpha}| \quad \mbox{and} \quad C_2 = \sum_{\beta} c_{\beta} |f_{\beta} \rangle \langle f_{\beta} | \, .
\end{equation}
We have then that
\begin{equation}
\left\| \rho^{\tilde{T}_1} \right\|_1 = \left\| F(C_1 \otimes C_2) \right\|_1 =  \left\| C_1 \otimes C_2 \right\|_1 =  \left\| C_1 \right\|_1 \cdot \left\| C_2 \right\|_1 = \left( \sum_{\alpha}c_{\alpha} \right)^2 \, .
\end{equation}
The bound claimed in (\ref{eq:SNbd}) now follows from (\ref{eq:entropybd}).
\end{proof}

\end{document}